\documentclass[twocolumn,twoside]{IEEEtran}

\usepackage{times,graphicx,epic,eepic,epsfig,amsmath,latexsym,amssymb,verbatim,revsymb}

\usepackage{theorem}
\newtheorem{definition}{Definition}

\newtheorem{lemma}[definition]{Lemma}

\newtheorem{theorem}[definition]{Theorem}

\def\squareforqed{\hbox{\rlap{$\sqcap$}$\sqcup$}}
\def\qed{\ifmmode\squareforqed\else{\unskip\nobreak\hfil
\penalty50\hskip1em\null\nobreak\hfil\squareforqed
\parfillskip=0pt\finalhyphendemerits=0\endgraf}\fi}
\def\endenv{\ifmmode\;\else{\unskip\nobreak\hfil
\penalty50\hskip1em\null\nobreak\hfil\;
\parfillskip=0pt\finalhyphendemerits=0\endgraf}\fi}
\newenvironment{proof+}[1]{\noindent\hspace{2em}\textit{{Proof #1:~} }}{\hfill\QED}
\newenvironment{remark}{\noindent \textbf{{Remark~}}}{}

\mathchardef\ordinarycolon\mathcode`\:
\mathcode`\:=\string"8000
\def\vcentcolon{\mathrel{\mathop\ordinarycolon}}
\begingroup \catcode`\:=\active
  \lowercase{\endgroup
  \let :\vcentcolon
  }

\newcommand{\nc}{\newcommand}
\nc{\rnc}{\renewcommand}
\nc{\beq}{\begin{equation}}
\nc{\eeq}{{\end{equation}}}
\nc{\beqa}{\begin{eqnarray}}
\nc{\eeqa}{\end{eqnarray}}
\nc{\lbar}[1]{\overline{#1}}
\nc{\bra}[1]{\langle#1|}
\nc{\ket}[1]{|#1\rangle}
\nc{\ketbra}[2]{|#1\rangle\!\langle#2|}
\nc{\braket}[2]{\langle#1|#2\rangle}
\nc{\proj}[1]{| #1\rangle\!\langle #1 |}
\nc{\avg}[1]{\langle#1\rangle}
\nc{\smfrac}[2]{\mbox{$\frac{#1}{#2}$}}
\nc{\tr}{\operatorname{tr}}
\nc{\ox}{\otimes}
\nc{\dg}{\dagger}
\nc{\dn}{\downarrow}
\nc{\cA}{{\cal A}}
\nc{\cB}{{\cal B}}
\nc{\cC}{{\cal C}}
\nc{\cD}{{\cal D}}
\nc{\cE}{{\cal E}}
\nc{\cF}{{\cal F}}
\nc{\cG}{{\cal G}}
\nc{\cH}{{\cal H}}
\nc{\cI}{{\cal I}}
\nc{\cJ}{{\cal J}}
\nc{\cK}{{\cal K}}
\nc{\cL}{{\cal L}}
\nc{\cM}{{\cal M}}
\nc{\cN}{{\cal N}}
\nc{\cO}{{\cal O}}
\nc{\cP}{{\cal P}}
\nc{\cR}{{\cal R}}
\nc{\cS}{{\cal S}}
\nc{\cT}{{\cal T}}
\nc{\cX}{{\cal X}}
\nc{\cZ}{{\cal Z}}
\nc{\rank}{\operatorname{rank}}
\nc{\rar}{\rightarrow}
\nc{\lrar}{\longrightarrow}
\nc{\polylog}{\operatorname{polylog}}
\nc{\1}{{\openone}}

\def\a{\alpha}
\def\b{\beta}
\def\g{\gamma}
\def\d{\delta}
\def\e{\epsilon}

\def\l{\lambda}
\def\m{\mu}

\def\r{\rho}
\def\s{\sigma}

\def\ph{\varphi}

\def\ps{\psi}
\def\o{\omega}

\def\G{\Gamma}

\def\Ph{\Phi}
\def\Ps{\Psi}
\def\O{\Omega}

\nc{\RR}{{{\mathbb R}}}
\nc{\CC}{{{\mathbb C}}}
\nc{\FF}{{{\mathbb F}}}
\nc{\NN}{{{\mathbb N}}}
\nc{\ZZ}{{{\mathbb Z}}}
\nc{\PP}{{{\mathbb P}}}
\nc{\QQ}{{{\mathbb Q}}}
\nc{\UU}{{{\mathbb U}}}
\nc{\EE}{{{\mathbb E}}}
\nc{\id}{{\operatorname{id}}}

\nc{\be}{\begin{equation}}
\nc{\ee}{{\end{equation}}}
\nc{\bea}{\begin{eqnarray}}
\nc{\eea}{\end{eqnarray}}
\nc{\<}{\langle}
\rnc{\>}{\rangle}
\nc{\Hom}[2]{\mbox{Hom}(\CC^{#1},\CC^{#2})}
\nc{\rU}{\mbox{U}}

\nc{\ob}[1]{#1}

\begin{document}

\title{Weak Decoupling Duality \protect\\ and Quantum Identification\thanks{24 October 2001. A preliminary version of this paper was presented as a contributed talk at the 12th QIP workshop, Santa Fe (NM), 12-16 January 2009.\protect\\ \indent PH is with the School of Computer Science, McGill University, Montreal, Canada. He was supported by the Canada Research Chairs program, the Perimeter Institute, CIFAR, FQRNT's INTRIQ, MITACS, NSERC, ONR through grant N000140811249 and QuantumWorks. Email: patrick@cs.mcgill.ca.\protect\\ \indent AW is with the Department of Mathematics, University of Bristol, Bristol BS8 1TW, U.K. and the Centre for Quantum Technologies, National University of Singapore, 2 Science Drive 3, Singapore 117542. He was supported through an Advanced Research Fellowship of the U.K. EPSRC, the EPSRC's ``QIP IRC'', the European Commission IP ``QAP'', by a Wolfson Research Merit Award of the Royal Society, a Philip Leverhulme Prize and an ERC Advanced Grant. The Centre for Quantum Technologies is funded by the Singapore Ministry of Education and the National Research Foundation as part of the Research Centres of Excellence programme. Email: a.j.winter@bris.ac.uk.}}

\author{Patrick Hayden, \textit{Member, IEEE} and Andreas Winter%
\vspace{3mm}\begin{center}{\sc Dedicated to the memory of Rudolf Ahlswede}\end{center}\vspace{3mm}}

\date{24 October 2011}

\maketitle

\begin{abstract}
If a quantum system is subject to noise, it is possible to perform
quantum error correction reversing the action of the noise if and
only if no information about the system's quantum state leaks to the
environment. In this article, we develop an analogous duality in the
case that the environment approximately forgets the identity of the
quantum state, a weaker condition satisfied by $\e$-randomizing
maps and approximate unitary designs. Specifically, we show that the environment approximately
forgets quantum states if and only if the original channel
approximately preserves pairwise fidelities of pure inputs, an
observation we call weak decoupling duality. Using this tool, we
then go on to study the task of using the output of a channel to
simulate restricted classes of measurements on a space of input
states.  The case of simulating measurements that test whether the
input state is an arbitrary pure state is known as equality testing
or  quantum identification. An immediate consequence of weak decoupling duality
is that the ability to perform quantum identification cannot be cloned.
We furthermore establish that the optimal amortized
rate at which quantum states can be identified through a noisy
quantum channel is equal to the entanglement-assisted classical
capacity of the channel, despite the fact that the task is quantum,
not classical, and entanglement-assistance is not allowed. In
particular, this rate is strictly positive for every non-constant
quantum channel, including classical channels.
\end{abstract}

\section{Introduction} \label{sec:intro}

Quantum channels in modern quantum information
theory~\cite{BennettShor} are modeled as completely positive and
trace-preserving maps $\cN:\cS(A) \rar \cS(B)$ between the state
spaces of quantum systems with Hilbert spaces $A$ and $B$. The
requirement of \emph{complete} positivity means that $\cN$ is not
just \emph{positive}, mapping positive semidefinite operators to
positive semidefinite operators, but that $\id\ox\cN$ is positive
for the identity map $\id$ on any $\cS(R)$. This distinction plays a
central role in the geometry of entanglement because positive but
not completely positive maps can be used to identify entangled
quantum states~\cite{e-witnesses}. This paper will take as its
starting point a similar observation about channel norms.

The Stinespring dilation theorem establishes a fundamental property
of quantum channels: for every channel $\cN$ there exists an ancilla
space $E$ and an isometry $V:A \hookrightarrow B\ox E$ such that
$\cN(\rho) = \tr_E V\rho V^\dagger$~\cite{Stinespring}. This means
that quantum noise can always be interpreted as information loss in
an otherwise deterministic evolution. Since $E$ and $V$ are
essentially unique (up to unitary equivalence), each channel $\cN$
also has an associated \emph{complementary channel} $\cN^c: \cS(A)
\rar \cS(E)$, with $\cN^c(\rho) = \tr_B V\rho V^\dagger$, which is
uniquely defined up to coordinate changes of $E$.

In quantum Shannon theoretic error correction we try to find two
channels $\cE$ and $\cD$ (an encoder and decoder) such that
$\cD\circ\cN\circ\cE \approx \id$. For now we shall consider the
encoding $\cE$ fixed, so that $\cN\circ\cE$ can be treated as a
single channel. The central insight of quantum error
correction~\cite{info.qec,SchumacherWestmoreland,Devetak:Q,KretschmannWerner}
is that the existence of a decoding operation $\cD$ for a channel
$\cN$, i.e.
\begin{equation}
  \label{eq:completely-corrected}
  \forall\rho\in\cS({RA}) \quad \bigl\|
    (\id\ox\cD \circ \cN)\rho^{RA} - \rho^{RA}
  \bigr\|_1
\leq \epsilon,
\end{equation}
is equivalent to the complementary channel being \emph{completely
forgetful}: for all Hilbert spaces $R$,
\begin{equation}
  \label{eq:completely.forget}
  \forall\rho,\sigma\in\cS({RA}) \quad \bigl\|
    (\id\ox\cN^c)\rho^{RA} - (\id\ox\cN^c)\sigma^{RA}
  \bigr\|_1 \leq \delta,
\end{equation}
with a universal relation between $\epsilon$ and $\delta$.

Here we determine a matching duality for the weaker property
of the complementary channel being only (approximately) \emph{forgetful}:
\begin{equation}
  \label{eq:forget}
  \forall\rho,\sigma\in\cS(A) \quad \bigl\|
    \cN^c(\rho^A) - \cN^c(\sigma^A)
  \bigr\|_1 \leq \delta.
\end{equation}
That this is a much weaker property was noticed in the contexts of
approximate encryption and remote state preparation~\cite{rand,rsp}.
The difference between Eqs.~(\ref{eq:completely.forget}) and
(\ref{eq:forget}) is precisely the difference between two norms on
superoperators, the na\"ive one inherited from the trace norm, and
the so-called completely bounded
norm~\cite{Paulsen:CB,Kitaev,KretschmannWerner}.
Not surprisingly, Eq.~(\ref{eq:forget}) will hold provided the main
channel approximately preserves the pairwise fidelities between
input pure states, a property we call \emph{geometry preservation}:
\begin{equation}
  \label{eq:geometry-preserving}
  \forall \ket{\psi},\ket{\ph}\in A
       \quad \bigl| \left\| \ph - \ps \right\|_1
            - \left\| \cN(\ph) - \cN(\ps) \right\|_1 \bigr|
       \leq \e.
\end{equation}

In fact, the reverse is also true. Our investigations will revolve
around \emph{weak decoupling duality}, which asserts that a channel
$\cN$ is geometry-preserving if and only if its complement $\cN^c$
is approximately forgetful, with dimension-independent functions
relating $\d$ and $\e$. Thus, an isometry with two outputs can
preserve geometry to at most one of them. Symmetrically, the
isometry can be forgetful to at most one output. 

%
The geometry preservation property, though much weaker than
transmission of quantum information, must nonetheless be considered
a way of preserving coherence: by virtue of weak decoupling duality, 
geometry preservation cannot be cloned. Indeed, if a
channel has multiple outputs, one of which is geometry-preserving,
then the rest must be forgetful.

Via weak decoupling duality, the many known examples of
approximately forgetful channels that are not completely forgetful
also provide examples of geometry-preserving channels that are not
correctable~\cite{rand,rsp,Ambainis-Smith,generic,Nayak,Harrow,Gross-Eisert,Aubrun}. 
Most strikingly, it is possible to preserve geometry
while almost halving the number of qubits from input to
output~\cite{winter:q-ID-1}. In that case, the geometry of the unit
sphere in $A$ is necessarily encoded into the eigenvectors
\emph{and} eigenvalues of the much smaller output state on $B$. In
contrast to quantum error correction, dimension counting reveals the
mixedness of the output state to be crucial to preserving the
geometry. Some of the geometry of the input state space of pure
quantum states is thus faithfully encoded as noise in the output
state.

Moreover, the analogy with the quantum error correction duality can
be made much stronger. There is a channel communication task very
similar to quantum state transmission which is intimately related to
geometry preservation: 
\emph{quantum identification}~\cite{winter:q-ID-1,winter:q-ID-2}.

\medskip\noindent
{Quantum identification} is a cooperative communication game between
two parties -- conventionally called Alice and Bob -- where Alice
has a given quantum state that she encodes in some way into the
channel, and Bob only wants to simulate measurements consisting of
an arbitrary pure state projector and its complement, which can
interpreted as performing the experiment asking ``Is this the
state?''~\cite{winter:q-ID-1}. The idea is that Alice has an
encoding channel $\cE$ and Bob has, for every pure state $\ph$, a
POVM $(D_\ph,\1-D_\ph)$ such that
\begin{equation}
  \label{eq:q.ID}
  \forall \ket{\psi},\ket{\ph} \quad
     \Bigl| \tr \bigl( (\cN\circ\cE)\psi \bigr) D_\ph
            - \tr \psi\ph \Bigr|
     \leq \epsilon.
\end{equation}
Such an object is called an \emph{$\epsilon$-quantum-ID code}.
(The name is adapted from the classical case~\cite{AhlswedeDueck:ID.A,AhlswedeDueck:ID.B}.
Indeed, in~\cite{Loeber:PhD,AhlswedeWinter:ID-q} the Ahlswede-Dueck theory
of identification is studied in the context of quantum channels;
both papers define ``quantum identification codes'', which however,
in the light of the above definition and~\cite{winter:q-ID-1,winter:q-ID-2},
are better named ``(classical) identification codes via quantum channels''.)

Note that Bob has at his disposal various quantum measurements at
the output of the channel, but the quality of
the code is measured by how well the statistics of this measurement
approximate the statistics of the ideal measurement he wants to
perform on the message state. While it may seem that this is an odd
way of defining a quantum communication task, normal quantum error
correction can also be described this way; namely, Bob wants to be
able to simulate \emph{all} measurements on the message state.
Clearly, if he can perform quantum error correction in the usual
sense, then he can perform the simulation. But conversely, it
follows from the methods
of~\cite{ChristandlWinter,HSW-gaussian,Renes10} that if he only has
two measurements approximating generalized $X$ and $Z$ observables
sufficiently well, he can build a quantum error correction procedure
$\cD$. Moreover, a quantum-ID code with $\e=0$ is itself a quantum
error correcting code; there is no difference between error
correction and identification if both tasks are to be performed
perfectly. But as we shall see, in the regime of non-zero error, 
$\e \neq 0$, the two concepts diverge.
Even the task of transmitting classical information is
conveniently reflected in the framework of simulating measurements:
In that case, Bob only wants to simulate the measurement of 
the generalized $Z$ observable.

With this, one can define in the usual way a \emph{quantum-ID
capacity} $Q_{\rm ID}(\cN)$ of many uses of the channel as the
highest rate at which qubits can be encoded and decoded as in
Eq.~(\ref{eq:q.ID}) with vanishing error -- see
Section~\ref{sec:qid} for details. Previously it was only known that
for the noiseless qubit channel $\id_2$, $Q_{\rm ID}(\id_2) = 2$,
double the value of both the the quantum and classical transmission
capacities~\cite{winter:q-ID-1}.

While reasoning directly about quantum identification (quantum-ID)
codes has
proved challenging, the duality between geometry preservation and
approximate forgetfulness provides a new approach to studying them.
Up to some technical conditions, geometry preservation is equivalent
to the existence of a quantum-ID code. It is therefore
possible to construct quantum-ID codes by finding
approximately forgetful maps. This approach is fruitful because
destroying information is a comparatively indiscriminate task.
Indeed, the analogous strategy has led to a number of
straightforward proofs of the hashing bound on the quantum capacity
of a quantum channel~\cite{HHYW,HSW-gaussian,klesse,merging}.
Classical data is not immune to analysis by purification either. The
duality between privacy amplification and data compression with
quantum side information has recently led to a proof in this
spirit~\cite{Renes10b,RenesRenner10} of the Holevo-Schumacher-Westmoreland theorem
on the classical capacity of a quantum
channel~\cite{Holevo98,SchumacherWestmoreland} .

With weak decoupling duality in hand, it is even possible to
establish a simple formula for an amortized version of the quantum
identification capacity; it is exactly equal to the
entanglement-assisted classical capacity of a quantum channel.

\subsection{Structure of the paper} \label{subsec:structure}

Section~\ref{sec:fidelity.alternative} contains the formal statement
and proof of the weak decoupling duality. The duality is studied in
more detail in Section~\ref{sec:qid}, where forgetfulness is shown
to be nearly equivalent to quantum identification. In that section
we provide a simple statement whose proof eliminates many technical
difficulties, as well as a more flexible version that we prove from
first principles. Section~\ref{sec:qid.cap} uses the flexible
version of the equivalence to construct quantum-ID codes for 
memoryless quantum channels. Section~\ref{sec:amortization-rate}
explores how much side communication is required to achieve the 
amortized quantum identification capacity, establishing that for
some channels, a positive rate is necessary.

\subsection{Notation} \label{subsec:notation}

We will restrict our attention throughout to finite dimensional
Hilbert spaces. If $A$ is a Hilbert space, we write $\cS(A)$ for the
set of density operators acting on $A$. Also, if $A$ and $B$ are two
finite dimensional Hilbert spaces, we write $AB \equiv A\otimes B$
for their tensor product. The Hilbert spaces on which linear
operators act will be denoted by a superscript.  For instance, we
write $\ph^{AB}$ for a density operator on $AB$. Partial traces will
be abbreviated by omitting superscripts, such as $\ph^A \equiv
\tr_B\ph^{AB}$.  We use a similar notation for pure states, e.g.\
$\ket{\psi}^{AB}\in AB$, while abbreviating $\psi^{AB} \equiv
\proj{\psi}^{AB}$. We will write $\id_A$ for the identity map on
$\cS(A)$ and $\id_2$ for the identity qubit channel. The symbol
$\1^A$ will be reserved for the identity matrix acting on the
Hilbert space $A$ and $\pi^A = \1^A / |A|$ for the maximally
mixed state on $A$ (where we denote by $|A|$ the dimension of
the Hilbert space $A$).

The trace norm of an operator, $\|X\|_1$ is defined to be $\tr|X| =
\tr\sqrt{X^\dagger X}$. The similarity of two density operators
$\ph$ and $\ps$ can be measured by the \emph{trace distance}
$\smfrac{1}{2} \| \ph - \ps \|_1$, which is equal to the maximum
over all possible measurements of the variational distance between
the outcome probabilities for the two states. The trace distance is
zero for identical states and one for perfectly distinguishable
states.

A complementary measure is the mixed state fidelity
\begin{equation} \label{eq:fidelity.defn}
  F(\ph,\ps)
  = \left\| \sqrt{\ph}\sqrt{\ps} \right\|_1^2
  = \left( \tr\sqrt{\sqrt{\ph}\psi\sqrt{\ph}} \right)^2,
\end{equation}
defined such that when one of the states is pure, $F(\ph,\ps) =
\tr\ph\ps$. More generally, the fidelity is equal to one for
identical states and zero for perfectly distinguishable states. We
will make frequent use of the following fundamental inequality
between fidelity and trace distance of
states~\cite[Prop.~5]{FuchsVandegraaf:fidelity}:
\begin{equation} \label{eq:F.vs.D}
  1 - \sqrt{F(\ph,\ps)}
  \leq \frac{1}{2}\| \ph-\ps \|_1
  \leq \sqrt{1-F(\ph,\ps)}.
\end{equation}
Both measures can be extended to unnormalized states, but
Eq.~(\ref{eq:F.vs.D}) need not hold in that case. Further properties
of the distance measures are collected in the Appendix.

\section{Weak decoupling duality} \label{sec:fidelity.alternative}

Our investigations will revolve around the duality between geometry
preservation and approximate forgetfulness, which we call weak decoupling duality. The rigorous statement is as follows:
%
%
\begin{theorem}[Weak decoupling duality] \label{thm:fidelity.alternative}
Let $\cN : \cS(A) \rar \cS(B)$ be a quantum channel with
complementary channel $\cN^c : \cS(A) \rar \cS(E)$. Approximate
geometry preservation on $B$ implies approximate forgetfulness for
$E$. That is,
\begin{align*}
    \forall \ket{\psi},\ket{\ph}\in A
      &\quad \| \ph - \ps \|_1 - \| \cN(\ph) - \cN(\ps) \|_1 \leq \d \\
    \text{implies } \forall \ket{\psi},\ket{\ph}\in A
      &\quad \| \cN^c(\ph) - \cN^c(\ps) \|_1 \leq 4\sqrt{2} \d^{1/4}.
\end{align*}
Conversely, approximate forgetfulness for $E$ implies approximate
geometry preservation on $B$:
\begin{align*}
    \forall \ket{\psi},\ket{\ph}\in A
      &\quad \| \cN^c(\ph) - \cN^c(\ps) \|_1 \leq \e \text{ implies }  \\
    \forall \ket{\psi},\ket{\ph}\in A
      &\quad \| \ph - \ps \|_1 - \| \cN(\ph) - \cN(\ps) \|_1 \leq 4\sqrt{2\e}.
\end{align*}
\end{theorem}
Note that we have dropped an absolute value sign as compared to
Eq.~(\ref{eq:geometry-preserving}) since $\|\ph - \ps\|_1 \geq
\|\cN(\ph) - \cN(\ps) \|_1$ holds automatically for all quantum
channels $\cN$. (See, for example, \cite{mikeandike}.)

The duality is a straightforward consequence of two basic results in
quantum information theory. The first is that the ability to
transmit classical data in two conjugate bases is equivalent to the
ability to transmit entanglement. That observation is the basis for
the stabilizer approach to quantum error correcting
codes~\cite{gottesman.stabilizer}. Here we will use a clean
approximate formulation due to Renes~\cite{Renes10}. The second
result is the continuity of the Stinespring dilation of a quantum
channel, established by Kretschmann \emph{et al.}~\cite{KretschmannWerner}.
Here we only need a corollary, which
can be interpreted as a bound on the information-disturbance
trade-off. The theorem is stated in terms of the following norms:
%
%
\begin{definition}
For a linear superoperator $\G : \cS(A) \rar \cS(B)$, let
\begin{equation*}
    \| \G \|_\diamond^{(k)}
    = \max_{\|X\|_1 \leq 1}
        \bigl\| (\id_k \ox \G)X \bigr\|_1,
\end{equation*}
where maximization is over operators $X$ on $\CC^k \ox A$.
Define $\| \G \|_\diamond = \sup_k \| \G \|_\diamond^{(k)}$, 
the \emph{completely bounded trace norm}~\cite{Paulsen:CB}
(also known as \emph{diamond norm}~\cite{Kitaev}).
\end{definition}
Note that the convexity of the trace norm implies that the supremum
is achieved on a rank-one operator (if $\G$ is Hermitian-preserving,
then on a pure quantum state). Since any operator on $A$ can be
``purified'' by a system of dimension $|A|$, it follows that the
supremum is achieved when $k = |A|$.

Of course, since all our Hilbert spaces are finite-dimensional,
all these norms are equivalent -- indeed, by Lemma~\ref{lem:norm.bound} 
in the Appendix,
\[
  \| \G \|_\diamond^{(1)} \leq \| \G \|_\diamond^{(k)} \leq k \| \G \|_\diamond^{(1)}.
\]
Since the factor of $k$ cannot be improved, this means that the
norms can differ by a factor as large as the dimension of $A$,
rendering the norms inequivalent in asymptotic settings, such as 
will be considered in the following. This can also be seen in 
the difference between approximately and completely forgetful
maps. There, $\G$ is the difference between a completely positive,
trace-preserving map and a constant map (on states); approximate
forgetfulness postulates a bound on $\|\G\|_\diamond^{(1)}$
while complete forgetfulness requires bounding $\|\G\|_\diamond$.

%
%
\begin{theorem}[Information-disturbance~\cite{KretschmannWerner}]
\label{thm:info.disturb} Let $V : A \rar B \ox E$ be an isometric
extension of the channel $\cN : \cS(A) \rar \cS(B)$ and let $\cN^c :
\cS(A) \rar \cS(E)$ be the complementary channel. Fix a state $\r
\in \cS(A)$ and let $\cR : \cS(A) \rar \cS(E)$ be the channel taking
all inputs to $\cN^c(\r)$. Then
\begin{equation*}
    \frac{1}{4} \inf_{\cD} \left\| \cD \circ \cN - \id \right\|_\diamond^2
    \leq \left\| \cN^c - \cR \right\|_\diamond
    \leq 2 \inf_{\cD} \left\| \cD \circ \cN - \id
    \right\|_\diamond^{1/2}.
\end{equation*}
Both infimums are over all quantum channels.
\end{theorem}
The proof of weak decoupling duality is a fairly routine matter of
combining these results:

\medskip
%
%
\begin{proof+}{of Theorem~\ref{thm:fidelity.alternative}}
We begin by assuming approximate geometry preservation. Fix
$\ket{\ph} \perp \ket{\ps}$ in $A$ then set $T =
\operatorname{span}(\ket{\ph},\ket{\ps})$. Suppose that
\begin{equation*}
    \| \cN(\o) - \cN(\xi) \|_1 \geq \| \o - \xi \|_1 - \d
\end{equation*}
for all $\ket{\o}, \ket{\xi} \in A$. Then if $\ket{\chi_\pm} =
\smfrac{1}{\sqrt{2}}( \ket{\ph} \pm \ket{\ps} )$, we have
\begin{align*}
    \left\| \cN( \ph ) - \cN( \ps ) \right\|_1 &\geq 2 - \d
        \quad \mbox{and} \\
    \left\| \cN( \chi_+ ) - \cN( \chi_- ) \right\|_1 &\geq 2 - \d.
\end{align*}
We can therefore transmit data in two conjugate bases through $\cN$,
which implies that entanglement is also faithfully transmitted. In
particular~\cite[Thm.~1]{Renes10} (with ``guessing
probability'' $1-\d/2$) implies that there exists a channel $\cD :
\cS(B) \rar \cS(T)$ such that
\begin{equation*}
    \left\| (\id_2 \ox \cD \circ \cN)\Ph - \Ph \right\|_1
    \leq 2\sqrt{\d},
\end{equation*}
where $\ket{\Ph} = \smfrac{1}{\sqrt{2}}(\ket{0}\ket{\ph} +
\ket{1}\ket{\ps})$. But trace norm monotonicity with respect to
dephasing the first system then gives
\begin{align*}
    &\left\| (\id_2 \ox \cD \circ \cN)\Ph - \Ph \right\|_1 \\
    &\geq \frac{1}{2} \big\|
        \proj{0} \ox \left[ (\cD \circ \cN )\ph - \ph \right] \\
     &\phantom\quad  \quad \quad \quad + \proj{1} \ox \left[ (\cD \circ \cN )\ps - \ps \right]
    \big\|_1 \\
    &=
    \frac{1}{2} \| (\cD \circ \cN)\ph - \ph \|_1
    + \frac{1}{2} \| (\cD \circ \cN)\ps - \ps \|_1.
\end{align*}
Therefore, $\| (\cD \circ \cN)\ph - \ph \|_1 \leq 4\sqrt{\d}$ and by changing the choice of dephasing basis, we can conclude that $\| \cD \circ \cN - \id_2 \|_\diamond^{(1)} \leq 4 \sqrt{\d}$.
Combining this with Lemma~\ref{lem:norm.bound} in the Appendix
implies that $\| \cD \circ \cN - \id_2 \|_\diamond \leq 8\sqrt{\d}$. The
information-disturbance theorem (Theorem \ref{thm:info.disturb}) applied with
$\cR$ the map taking all states to $\cN^c(\ph)$
then implies that for all $\ket{\o} \in T$,
\begin{equation*}
    \| \cN^c(\ph)-\cN^c(\o) \|_1
    \leq 2 (8\sqrt{\d})^{1/2}
    = 4 \sqrt{2}\d^{1/4}.
\end{equation*}
Since $T$ is an arbitrary two-dimensional subspace of $A$, however,
the inequality must hold for all $\ket{\ph}$ and $\ket{\o}$ in $A$.

For the converse, suppose that, for all states $\ket{\ph}, \ket{\ps}
\in A$, the inequality $\| \cN^c(\ph) - \cN^c(\ps) \|_1 \leq \e$
holds. Fix $\ket{\ph}$ and $\ket{\ps}$ then let $\tilde\cN^c$ be the
restriction of $\cN^c$ to states on $T =
\operatorname{span}(\ket{\ph},\ket{\ps})$. Let $\cR$ be the channel
on $\cS(T)$ that always outputs $\cN^c(\ps)$. Then once more by
Lemma~\ref{lem:norm.bound} in the Appendix,
$\| \tilde\cN^c - \cR \|_\diamond \leq 2
\e$. Using this time the lower bound from Theorem~\ref{thm:info.disturb}, there exists a channel $\cD
: \cS(B) \rar \cS(T)$ such that $\smfrac{1}{4} \| \cD \circ \cN - \id \|_\diamond^2 \leq 2\epsilon$. In particular, for all $\ket{\o} \in T$,
\begin{equation*}
    \frac{1}{4} \left\|
        (\cD \circ \cN)\o - \o
    \right\|_1^2
    \leq 2\e.
\end{equation*}
Applying the triangle inequality several more times gives:
\begin{align*}
    4\sqrt{2\e}
    &\geq \left\| (\cD \circ \cN)\ph - \ph \right\|_1
            + \left\| (\cD \circ \cN)\ps - \ps \right\|_1 \\
    &\geq \left\| \ph - \ps \right\|_1
            - \left\| (\cD \circ \cN)(\ph - \ps) \right\|_1 \\
    &\geq \left\| \ph - \ps \right\|_1
            - \left\| \cN(\ph - \ps) \right\|_1,
\end{align*}
where the final inequality used that the quantum channel $\cD$
cannot increase the trace norm. Rearranging the final expression
gives the desired inequality.
\end{proof+}
%
%
%

\section{Quantum identification} \label{sec:qid}

Quantum identification allows a sender to transmit arbitrary quantum
states but only allows the receiver to perform a restricted set of
measurements, namely tests to determine whether the transmitted
state consists of an arbitrary target state. The receiver gets to
choose the target state \emph{after} the sender has transmitted, so
the code must work for all targets. If the test can be performed
perfectly, then quantum identification is easily seen to be
equivalent to quantum state transmission, but in the approximate
setting, the tasks are not equivalent.
%
%
\begin{definition} \label{defn:qid.code}
{\bf \cite{winter:q-ID-1}}
An \emph{$\e$-quantum-ID code} for the channel $\cN : \cS(A) \rar \cS(B)$
consists of an encoding map $\cE : \cS(S) \rar \cS(A)$ and, for
every pure state $\ket{\ph} \in S$, a POVM $(D_\ph,\1-D_\ph)$ acting
on $\cS(B)$ such that
\begin{equation*}
    \forall \ket{\ps},\ket{\ph} \in S
    \quad
    \Big| \tr\big( (\cN \circ \cE) \ps \big) D_\ph - | \braket{\ph}{\ps} |^2 \Big|
    \leq \e.
\end{equation*}
\end{definition}
If the receiver had been able to perform the measurement
$(\proj{\ph},\1-\proj{\ph})$ on the input state $\ket{\ps}$, then he
would have observed outcome $\proj{\ph}$ with probability
$|\braket{\ph}{\ps}|^2$. The definition therefore ensures that the
receiver can simulate the measurement for all input and target
states.

Many variants of the definition have been proposed. In particular,
one could imagine drawing a distinction between oblivious ID codes,
in which the sender is only given a physical quantum state to send,
and visible ID codes, in which the sender knows the identity of the
state she is trying to transmit~\cite{winter:q-ID-1}. Entanglement
assistance is also interesting and exceptionally powerful in the
visible setting~\cite{remote-state-preparation}.  A different task
that is nonetheless similar in spirit is to use quantum states as
``fingerprints'' for identifying classical messages in a model where
pairs of messages are to be compared by a
referee~\cite{fingerprinting}. For comparing quantum states,
however, the simple definition considered here is arguably the most
natural.

If we integrate the encoding $\cE$ and noisy channel $\cN$ from
Definition~\ref{defn:qid.code} into a single map with output $B$ and
environment $E$, we may think of the code Hilbert space $S$ as a
subspace of $B\ox E$. More formally, if we let $V$ be the
Stinespring dilation of $\cN \circ \cE$, then $V: S \hookrightarrow
B \ox E$ and we can identify the code with a subspace of $B \ox E$.
This identification simplifies the notation and we will use it for
the remainder of the paper.

The main result of this section is a demonstration that a subspace
of $B \ox E$ is a quantum-ID code for $B$ iff it is approximately
forgetful for $E$. (There is a small technical caveat to the
statement: the reduced states on $E$ must also obey a regularity
condition for the reverse implication to hold, but we will defer
discussion of the details.) For the moment, let us begin by
considering the relationship between quantum identification and geometry
preservation.
%
%
\begin{lemma} \label{lemqID.fidelity}
Let $S \subseteq B \ox E$ be a subspace of a tensor product Hilbert
space that is an $\epsilon$-quantum-ID code for $B$. In other words,
suppose that, for each pure state $\ket{\ph}\in S$, there exists an
operator $0 \leq D_\ph \leq \1$ on $B$ such that for all pure states
$\ket{\ph},\ket{\psi} \in S$,
\[
    \bigl| \tr \psi^B D_\ph - \tr \psi\ph \bigr| \leq \epsilon.
\]
Then, for all $\ket{\ph},\ket{\psi} \in S$,
\begin{equation*}
F(\ph,\ps) \leq F(\ph^B,\ps^B) \leq F(\ph,\ps) + 4 \sqrt{\e}.
\end{equation*}
\end{lemma}
\begin{proof}
Consider the measurement $(D_\ph,\1-D_\ph)$ and associated channel
$M: \rho \mapsto \operatorname{diag}( \, \tr\rho D_\ph, 1-\tr\rho D_\ph \, )$ which acts on $\cS(B)$.
By applying the monotonicity of the fidelity under quantum channels
to $\tr_E$ and $M$, we get
\[\begin{split}
    F(\psi,\ph)  \leq F(\psi^B,\ph^B)
                    &\leq F\bigl( M(\psi^B),M(\ph^B) \bigr) \\
                    &\leq \left( \sqrt{\tr\psi^B D_\ph} + \sqrt{\epsilon} \right)^2 \\
                    &\leq F(\psi,\ph) + 2\sqrt{\epsilon} + \epsilon + \epsilon,
\end{split}\]
which proves the lemma.
\end{proof}

\medskip
The fidelity is therefore approximately preserved by quantum-ID
codes. Geometry preservation is defined in terms of the
trace distance, however, not the fidelity. While it is indeed the case that
quantum-ID codes preserve geometry, the argument is
somewhat more delicate because applying the measurement $(D_\ph, \1-D_\ph)$
causes a significant drop in the trace distance even as it leaves the fidelity
nearly unchanged. Instead, Theorem~\ref{thm:ID.dualto.rand} will
allow us to infer that quantum-ID codes preserve geometry by
virtue of the fact that their complementary channels are forgetful.

In order to succeed at quantum identification, the following lemma
demonstrates that it is sufficient to be able to identify orthogonal
states:
%
%
\begin{lemma} \label{lem:ortho.to.all}
Let $S \subseteq B \ox E$ be a subspace of a tensor product Hilbert
space such that for $\ket{\ph} \in S$ there exists $0 \leq D_\ph
\leq \1$ acting on $B$ satisfying
\begin{align*}
    \tr \ph^B D_\ph \geq 1 - \d
    \quad \mbox{and} \quad
    \tr \ps^B D_\ph < \d
\end{align*}
whenever $\ket{\ps} \in S$ is orthogonal to $\ket{\ph}$. Then $S$ is
a quantum-ID code with error probability $\d +
2\sqrt{\d}$.
\end{lemma}
\begin{proof}
Let $\ket{\ph}, \ket{\ps} \in S$ be arbitrary and let $\ket{\ph'}$
be orthogonal to $\ket{\ph}$ in
$\operatorname{span}(\ket{\ph},\ket{\ps})$. Write $$\ket{\ps} = \a
\ket{\ph} + \b \ket{\ph'}.$$ 
Expanding shows that $\tr\psi^B D_\ph$ is equal to
\[\begin{split}
  |\alpha|^2 &\tr\ph^B D_\ph + |\beta|^2 \tr{\ph'}^B D_\ph \\
              &\phantom{=}
                +\alpha \overline\beta \tr\ket{\ph}\!\bra{\ph'}(D_\ph \ox\1) 
                +\overline\alpha \beta \tr\ket{\ph'}\!\bra{\ph}(D_\ph \ox\1),
\end{split}\]
which results in
\[\begin{split}
  \bigl| \tr\psi^B D_\ph &- |\alpha|^2 \bigr| \\
     &\leq |\alpha|^2 (1-\tr\ph^B D_\ph) + |\beta|^2 \tr{\ph'}^B D_\ph \\
     &\phantom{======}
                 + 2|\alpha\beta| |\bra{\ph} (D_\ph\ox\1) \ket{\ph'}| \\
     &\leq |\alpha|^2 (1-\tr\ph^B D_\ph) + |\beta|^2 \tr{\ph'}^B D_\ph \\
     &\phantom{======}         
                 + 2|\alpha\beta| \sqrt{\bra{\ph'} (D_\ph\ox\1) \ket{\ph'}} \\
     &\leq \d + 2 \sqrt{\d},
\end{split}\]
where we have used the Cauchy-Schwarz inequality and the assumption
that orthogonal states in $S$ can be well discriminated.
\end{proof}

\medskip
Now we are ready to state and prove our main  result on the duality
between quantum identification and approximate forgetfulness. As with weak decoupling duality, we have chosen to prove the theorem by composing general purpose results for the purpose of pedagogical clarity, which leads to artificially poor scaling of the parameters. Readers concerned with optimizing the parameters should also consult Theorem \ref{thm:ID.dualto.rand.tech}.
%
%
\begin{theorem}[Identification and forgetfulness] \label{thm:ID.dualto.rand}
Quantum-ID codes and forgetfulness are dual in the following
quantitative sense. If a subspace $S \subseteq B \ox E$ is an
$\e$-quantum-ID code for $B$, then $E$ is approximately
$\delta$-forgetful:
\[
    \forall \ket{\ph},\ket{\psi} \in S\quad
            \frac{1}{2}\big\| \ph^E-\psi^E \big\|_1
            \leq \delta:=7\sqrt[4]{\epsilon}.
\]
Conversely, if $E$ is approximately $\delta$-forgetful, then
geometry is approximately preserved on $B$:
\[
    \forall\ket{\ph},\ket{\psi} \in S\quad
        \big\| \ph - \ps \big\|_1 - \big\| \ph^B - \ps^B \big\|_1
        \leq \e := 4\sqrt{2\d}.
\]
If, in addition, the nonzero eigenvalues of $\ph^B$ lie in the
interval $[\m,\l]$ for all $\ket{\ph} \in S$, then $S$ is an
$\eta$-quantum-ID code for $\eta := 7\d^{1/8}\sqrt{\l/\m}$.
\end{theorem}
\begin{remark}
While it would be desirable to eliminate the eigenvalue condition at
the end of the theorem, the condition is fairly natural in this
context. If the reduced states $\ph^E$ are very close to a single
state $\s^E$ for all $\ket{\ph} \in S$, then all the $\ket{\ph}$ are
very close to being purifications of $\s^E$, meaning that they
differ from one another only by a unitary plus a small perturbation.
If $\s^E$ is the maximally mixed state or close to it, then the
assumption will be satisfied.
\end{remark}

\medskip
\begin{proof}
For the first part, recall that if $S$ is a quantum-ID code with
error probability $\e$, then for each pure state $\ket{\ph}\in S$
there exists an operator $0 \leq D_\ph \leq \1$ on $B$ such that for
all pure states $\ket{\ph},\ket{\psi} \in S$,
  \[
    \bigl| \tr \psi^B D_\ph - \tr \psi\ph \bigr| \leq \epsilon.
  \]
Just as in the proof of Theorem~\ref{thm:fidelity.alternative}, the
hypothesis implies that data can be transmitted in two conjugate
bases with guessing probability $1-\e$. Running exactly the same
argument as was made in that proof gives that for all $\ket{\ph},
\ket{\ps} \in S$,
\begin{equation}
    \frac{1}{2} \big\| \ph^E - \ps^E  \big\|_1
    \leq 4 \sqrt{2} (2\e)^{1/4} \leq 7 \e^{1/4}.
\end{equation}

The second part is just a restatement of one direction of the
weak decoupling duality, but it is a useful step on the way to the
third part, which is more challenging since it requires the
construction of the decoder, that is, the operators $D_\ph$.

Indeed, given $\ket{\ph} \in S$, and arbitrary $\ket{\psi} \perp
\ket{\ph}$ in $S$, we learn from the second part that
\begin{equation} \label{eq:qid.b.tracenorm}
    \| \ph^B - \ps^B \|_1 \geq 2 - 4\sqrt{2\d}.
\end{equation}
By Helstrom's theorem on the optimal discrimination of $\ph^B$ and
$\ps^B$~\cite{helstrom:discrimination}, there exists a projector
$P_{\ph,\ps}$ on $B$ such that
\begin{equation}
  \label{eq:good.discr}
  \tr\ph^B P_{\ph,\ps} \geq 1-2\sqrt{2\delta},\quad
  \tr\ps^B P_{\ph,\ps} \leq 2\sqrt{2\delta}.
\end{equation}
The problem with using $P_{\ph,\ps}$ as the decoding is that this
projector may indeed depend not only on $\ph$, but also on $\psi$.
Since the goal is to find a single projector that the receiver can use
to identify $\ph$ that will work regardless of whether the input is $\ph$ 
or $\ps$, that is unacceptable.
Still, let us confirm first that if we manage to find one effect
operator $D_\ph$ that can deal with all $\ps$ at once, then by
Lemma~\ref{lem:ortho.to.all} we'll be done. Our strategy for doing
so will be to first extend Eq.~(\ref{eq:good.discr}) to all mixed
states orthogonal to $\ket{\ph}$ and supported on $S$, and then use
a minimax argument to extract a single operator independent of
$\ps$.

Lemma~\ref{lemmixing.vs.fidelity} in the Appendix
can be used directly to see that
for all mixed states $\sigma$ supported on $S$ and orthogonal to
$\ph$,
\begin{align*}
    F(\ph^B,\s^B)
    \leq \frac{\l^2}{\mu^2}\max F(\ph^B,\ps^B)
    \leq 4 \sqrt{2\d} \frac{\l^2}{\mu^2},
\end{align*}
where the maximization is over all $\ket{\ps} \in S$ orthogonal to
$\ket{\ph}$ and the second inequality is an application of
Eq.~(\ref{eq:F.vs.D}) to Eq.~(\ref{eq:qid.b.tracenorm}). Applying
Eq.~(\ref{eq:F.vs.D}) a second time gives
\begin{equation*}
    \frac{1}{2} \big\| \ph^B - \ps^B \big\|_1
    \geq 1 - 2(2\d)^{1/4} \frac{\l}{\m}.
\end{equation*}
Applying Helstrom's theorem to $\ph^B$ and $\s^B$ yields a projector
$P_\sigma$ with
  \[
    \tr\ph^B P_\sigma - \tr\sigma^B P_\sigma
    \geq 1 - 2(2\d)^{1/4} \frac{\l}{\m}.
  \]
Von Neumann's minimax theorem then ensures the existence of a saddle
point in the following two-player game~\cite{vonNeumann:spiele}
(see Ky Fan~\cite{kyfan} for a more general version). One player
selects $0\leq P \leq \1$ while the other player selects a state
$\sigma$ supported on $S$ and orthogonal to $\ph$. The strategy
spaces are therefore closed and convex. The payoff function is
$1-\tr\ph^B P + \tr\sigma^B P$, which is linear in each argument.
Thus, the minimax theorem guarantees that there exists an operator
$0 \leq D_\ph \leq \1$ such that for all $\sigma$ supported on $S$
and orthogonal to $\ph$,
\begin{align*}
    \tr\ph^B D_\ph &\geq 1 - 2(2\d)^{1/4} \frac{\l}{\m}, \\
    \tr\sigma^B  D_\ph &\leq2(2\d)^{1/4} \frac{\l}{\m},
\end{align*}
and applying Lemma~\ref{lem:ortho.to.all} finishes the proof.
\end{proof}

\medskip
Unfortunately, Theorem \ref{thm:ID.dualto.rand} is not quite strong
enough to prove our main result on the quantum identification
capacity. To control the ratio of the largest to smallest
eigenvalues of the coding states, we need to act on them by typical
projectors that cause a slight distortion. To accomodate this
complication, we will instead use the following slightly more
flexible version of the converse that behaves better with respect to
the distortion. In particular, the amount of distortion enters the
bound on the quality of the quantum-ID code in a term independent of
the eigenvalue constraint. That separation proves to be crucial
because the eigenvalues cannot be controlled independently of the
distortion.
%
%
\begin{theorem} \label{thm:ID.dualto.rand.tech}
Let $S \subseteq B \ox E$ be a subspace and $0 \leq X  \leq \1$ an
operator acting on $B \ox E$ such that $\tr(X \ph X^\dg) \geq 1 -
\e$ for all $\ket{\ph} \in  S$. For any
state $\ket{\o} \in S$, write $\tilde{\o} = X \o X^\dg$. If there
exists a state $\O$ such that
\begin{equation*}
  \forall \ket{\ph} \in S \quad
    \big\| \tilde\O^E - \tilde{\ph}^E \big\|_1 \leq \d
\end{equation*}
with $0 \leq \d, \e \leq 1/15$
and, in addition, the nonzero eigenvalues of $\tilde\O^E$ lie in the
interval $[\m,\l]$, then $S$ is an $\eta$-quantum-ID code for $\eta
:= 3( 30 \l \d / \m + 3 \sqrt{\e} + 4 \d )^{1/2}$.
\end{theorem}
\begin{proof}
Let $\ket{\ph}$ and $\ket{\ps}$ be orthonormal states in $S$. We
will begin by showing that $\tilde\ph^B$ and $\tilde\ps^B$ can be
effectively distinguished. To this end, consider the states
\begin{align*}
  \ket{\vartheta_{\pm}} &= \frac{1}{\sqrt{2}}\ket{\ph} \pm \frac{1}{\sqrt{2}}\ket{\ps}, \\
  \ket{\chi_{\pm}}      &= \frac{1}{\sqrt{2}}\ket{\ph} \pm \frac{i}{\sqrt{2}}\ket{\ps},
\end{align*}
which form two orthogonal pairs. Then
\begin{align*}
  \tilde\vartheta^E_\pm &= \frac{1}{2}\tilde\ph^E + \frac{1}{2}{\tilde\ps}^E
    \pm \frac{1}{2}\bigl( \tr_B \ket{\tilde\ph}\!\bra{\tilde\ps} + \tr_B \ket{\tilde\ps}\!\bra{\tilde\ph} \bigr), \\
  \chi^E_\pm      &= \frac{1}{2}\tilde\ph^E + \frac{1}{2}{\tilde\ps}^E
    \mp \frac{i}{2}\bigl( \tr_B \ket{\tilde\ph}\!\bra{\tilde\ps} - \tr_B \ket{\tilde\ps}\!\bra{\tilde\ph} \bigr),
\end{align*}
and, by assumption,
\[
  \frac{1}{2}\| \tilde\vartheta^E_+ - \tilde\vartheta^E_- \|_1 \leq \d
     \quad\mbox{and}\quad
   \frac{1}{2}\| \tilde\chi^E_+ - \tilde\chi^E_- \|_1 \leq \d.
\]
Combining these relations reveals that $\| \tr_B
\ket{\tilde\ph}\!\bra{\tilde\ps} \pm \tr_B
\ket{\tilde\ps}\!\bra{\tilde\ph} \|_1 \leq 4\d$, hence by the
triangle inequality, $\| \tr_B \ket{\tilde\ph}\!\bra{\tilde\ps} \|_1
\leq 8\d$. But this gives us, by virtue of Lemma~\ref{lemnice},
\begin{equation} \label{eq:tilde.fidelity}
  F(\tilde\ph^B,{\tilde\ps}^B) \leq 64\delta^2.
\end{equation}
To proceed as in the proof of Theorem~\ref{thm:ID.dualto.rand}, we
need to show that any $\ket{\ph} \in S$ and mixed state $\s$
supported on the orthogonal complement of $\ket{\ph}$ in $S$ can
also be distinguished. In order to apply
Lemma~\ref{lemmixing.vs.fidelity} in the Appendix,
we will show that the largest and
smallest nonzero eigenvalues of $\ph^B$, or equivalently, $\ph^E$,
are well-behaved modulo a little bit of truncation. Indeed, let
$O = (O_j)$ and $p = (p_j)$ be the eigenvalues of $\tilde\O^E$ and
$\tilde\ph^E$, respectively, in nonincreasing order. Then
\begin{equation*}
    \big\| O - p \big\|_1 \leq \big\| \tilde\O^E - \tilde\ph^E \big\|_1 \leq \d.
\end{equation*}
Define the set
\begin{equation*}
    J = \big\{ j \, : \, (1-\g) p_j \leq O_j \leq (1+\g) p_j \big\}.
\end{equation*}
Then
\begin{equation*}
    \g \sum_{j \not\in J} p_j \leq \sum_{j \not\in J} |O_j - p_j| \leq \d,
\end{equation*}
implying that
\begin{equation*}
    \sum_{j \in J} p_j
    = \sum_j p_j - \sum_{j \not\in J} p_j
    \geq (1 - \e) -  \d/\g.
\end{equation*}
Fixing $\g = 1/2$ implies that for each $\ket{\ph} \in S$, there is
a positive semidefinite operator $\hat\ph^B \leq \tilde\ph^B$
satisfying $\tr \hat\ph^B \geq 1 - \e-2\d$ and whose eigenvalues lie
in the interval $[\m/2,3\l/2]$. 

Now let $\ket{\ph} \in S$ and consider any state $\s = \sum_i q_i
\ps_i$ whose support lies in the orthogonal complement of
$\ket{\ph}$ in $S$. Since the states $\ket{\psi_i}$ are in $S$, the truncation procedure of the previous paragraph
can be used to construct operators $\hat\psi_i$. Let $\hat\s = \sum_i q_i \hat\ps_i$. Then by
Lemma~\ref{lemmixing.vs.fidelity},
\begin{align*}
    F( \hat\ph^B, \hat\s^B )
    &\leq \frac{9\l^2}{\m^2} \max F(\hat\ph^B,\hat\ps^B ) \\
    &\leq \frac{9\l^2}{\m^2} \max F(\tilde\ph^B,\tilde\ps^B ) \\
    &\leq \frac{9\l^2}{\m^2} 64 \d^2
    = \frac{576\l^2 \d^2}{\m^2}.
\end{align*}
Both maximizations are over states $\ket{\ps} \in S$ such that
$\braket{\ph}{\ps} = 0$. The second inequality follows from the fact
that $\hat\ph^B \leq \tilde\ph^B$ (and likewise for $\ps$) along
with Lemma~\ref{lem:fidelity.op.mono} while the third arises by
substituting in the result of Eq.~(\ref{eq:tilde.fidelity}).
Introducing one last decoration for our states, let $\bar\ph^B =
\hat\ph^B / \tr \hat\ph^B$ and likewise for $\s$. Applying
Eq.~(\ref{eq:F.vs.D}) with attention paid to the fact that
$\hat\ph^B$ and $\hat\s^B$ are not normalized gives
\begin{equation*}
    \frac{1}{2} \big\| \bar\ph^B - \bar\s^B \big\|_1
    \geq 1 - \frac{24\l\d}{\m} \frac{1}{1-\e-2\d}
    \geq 1 - \frac{30\l\d}{\m},
\end{equation*}
where the final inequality uses that $\e \leq 1/15$. Applying
Helstrom's theorem to $\bar\ph^B$ and $\bar\s^B$ implies that there
exists a projector $P_\s$ such that
\begin{equation*}
    \tr\bar\ph^B P_\sigma - \tr\bar\sigma^B P_\sigma
    \geq 1 - \frac{30\l\d}{\m}.
\end{equation*}
Next we invoke von Neumann's minimax theorem, just as in the proof of
Theorem~\ref{thm:ID.dualto.rand}, for the payoff function $1 - \tr
\bar\ph^B P + \tr \hat\s^B P$, with the strategy space of the second
player the convex hull of the operators $\hat\ps^B$, where $\ket{\ps} \in
S$ ranges over states orthogonal to $\ket{\ph}$. (The operators $\hat\sigma^B$ are not normalized but that will not cause any difficulties.) This provides an
operator $0 \leq D_\ph \leq \1$ such that
\begin{align}
    \tr \bar\ph^B D_\ph &\geq 1 - \frac{30\l\d}{\m} \label{eq:truncated.game}
        \quad \mbox{and} \\
    \tr \hat\s^B D_\ph &\leq \frac{30\l\d}{\m}.
\end{align}
But
\begin{align*}
    \big| \tr \ph^B D_\ph &- \tr \bar\ph^B D_\ph \big| \\
    &\leq \| \ph^B - \bar\ph^B \|_1 \\
    &\leq \| \ph^B - \tilde\ph^B \|_1 +  \| \tilde\ph^B - \hat\ph^B \|_1 + \| \hat\ph^B - \bar\ph^B \|_1 \\
    &\leq 2 \sqrt{\e} + 2 \d  + \big| 1 - (1 - \e - 2\d) \big| \\
    &\leq 3 \sqrt{\e} + 4 \d,
\end{align*}
where the fourth line follows from the gentle measurement lemma
(Appendix, Lemma~\ref{lemma:gentle}), the definition of $\hat\ph^B$,
and the fact that $\bar\ph^B = \hat\ph^B /
(\tr\hat\ph^B)$. Similarly, for any $\hat\s^B = \sum_i q_i
\hat\ps_i^B$ a convex combination of states arising from
$\ket{\ps_i} \in S$ perpendicular to $\ket{\ph}$,
\begin{align*}
    \big| \tr \s^B D_\ph &- \tr \hat\s^B D_\ph \big| \\
    &\leq \| \s^B - \hat\s^B \|_1 \\
    &\leq \sum_i q_i \| \ps_i^B - \hat\ps_i^B \|_1 \\
    &\leq \sum_i q_i \left( \| \psi_i^B - \tilde\psi_i^B \|_1 + \| \tilde\psi_i^B - \hat\psi_i^B \|_1 \right) \\
    &\leq 2 \sqrt{\e} + 2 \d.
\end{align*}
Combining these estimates with the outcome of the minimax theorem in
Eq.~(\ref{eq:truncated.game}) and Lemma~\ref{lem:ortho.to.all}
completes the proof.
\end{proof}
%
%
%

\section{Quantum identification capacity} \label{sec:qid.cap}

While it might not be possible to design low error quantum-ID codes
for any given channel, the situation becomes more promising if many
uses of the channel are allowed. In analogy with classical and
quantum data transmission, we can define asymptotic quantum-ID codes
as follows.
%
%
\begin{definition}[Quantum-ID capacity~\cite{winter:q-ID-1}]
\label{defn:qid.am.code}
A rate $Q$ is said to be \emph{achievable for quantum identification 
over $\cN$} if for all $\e > 0$ and sufficiently large $n$, there
are $\e$-quantum-ID codes for $\cN^{\ox n}$ with encoding domain $S$
of dimension at least $2^{nQ}$. The \emph{quantum identification
capacity} $Q_{\rm ID}(\cN)$ is defined as the supremum of the
achievable rates.
\end{definition}

The capacity should be interpreted as the number of qubits that can
be identified per use of the channel $\cN$ in the limit of many uses
of the channel. The only nontrivial channel for which the quantum
identification capacity was known prior to this paper was the
identity channel: asymptotically, a noiseless qubit channel can be
used to identify two qubits. That is, $Q_{\rm ID}(\id_2) = 2$~\cite{winter:q-ID-1}.
As we will see below, the theory of the quantum identification capacity is
considerably simpler when the given channel $\cN$ can be used in
conjunction with noiseless channels to the receiver. This obviously
increases the capacity, so the interesting question is how much the
use of $\cN$ increases the quantum identification capacity over what
would have been achievable with the noiseless channels alone. When
defining the achievable amortized rates it is therefore necessary to
subtract off two qubits for every noiseless qubit channel used per
copy of $\cN$.
%
%
\begin{definition}[Amortized quantum-ID capacity] \label{defn:am.qid.code}
A rate $Q$ is said to be \emph{achievable for amortized quantum
identification over $\cN$} if for all $\e > 0$ and sufficiently
large $n$, there are $\e$-quantum-ID codes for $\id_C \ox \cN^{\ox
n}$ with encoding domain $S$ such that 
$Q \leq \smfrac{1}{n} (\log |S| - 2 \log |C|)$. where $\log=\log_2$ is
the binary logarithm throughout this paper. The \emph{amortized quantum
identification capacity} $Q_{\rm ID}^{\rm am}(\cN)$ is defined as
the supremum of the achievable rates.
\end{definition}
Readers familiar with the identification capacities of classical channels might be surprised to see that the dimension of a quantum-ID code scales only exponentially with the number of channel uses, as opposed to doubly exponentially. The essential difference between the classical and quantum settings is that the number of distinguishable quantum states in dimension $d$ already scales exponentially with $d$, which makes quantum identification a much more demanding task. Nonetheless, as we will see below, the amortized quantum identification capacity can be positive for some channels with zero quantum capacity, like the noiseless bit channel. One then finds that the dimension of the quantum-ID code \emph{can} scale super-exponentially with the number of qubits used to supplement the classical channel.

Weak decoupling duality is a very effective tool for studying the
quantum-ID capacities. As a warm-up, the fact that the complements
of quantum-ID codes are forgetful supplies a quick answer to an open
question from~\cite{winter:q-ID-1}:
%
%
\begin{theorem}
  \label{thm:Q.ID.zero}
  If $\cN$ is an antidegradable channel, that is, if there exists channel
  $\cT$ such that $\cN = \cT\circ\cN^c$, then $Q_{\rm ID}(\cN)=0$. This is
  true in particular for the noiseless cbit channel $\overline\id_2$.
  More generally, if the quantum capacity of the channel vanishes, $Q(\cN)=0$,
  then so does the quantum-ID capacity, $Q_{\rm ID}(\cN)=0$.
\end{theorem}
\begin{proof}
Given a quantum-ID code for the channel $\cN$ that
encodes as little as one qubit, the channel $\cN \circ \cE$ will be
geometry-preserving if $\cE$ is the encoding map. Hence, by weak decoupling duality, the channel complementary to
$\cN\circ\cE$ will be approximately forgetful. But if $\cN$ is antidegradable, then
so is $\cN \circ \cE$, meaning that the channel complementary to $\cN \circ \cE$ can
simulate $\cN \circ \cE$. But then the complementary channel would be simultaneously forgetful and
geometry-preserving, a contradiction.

For the more general statement, we show the contrapositive: assume
$Q_{\rm ID}(\cN) > 0$, then for all $\epsilon>0$ and sufficiently large
$n$, $\cN^{\ox n}$ has in particular a $2$-dimensional quantum-ID code $S$
which is $\epsilon$-close to being forgetful for the environment,
by Lemma~\ref{lemqID.fidelity}. But by Lemma~\ref{lem:norm.bound} this
means that the channel from the code qubit to the environment is
arbitrarily close to a constant map in the diamond norm. At this point
we can then invoke Theorem~\ref{thm:info.disturb} on 
information-disturbance~\cite{KretschmannWerner} to conclude that
the channel from the code qubit to $B^n$ can be arbitrarily well
error-corrected. (Note that this argument is following our proof of the
weak decoupling duality; in particular, any $2$-dimensional subspace
of a quantum-ID code, and in fact any subspace of sufficiently small dimension, is a quantum error-correcting code!) 
By the Lloyd-Shor-Devetak theorem on the quantum capacity 
(see~\cite{Devetak:Q}), this implies that there exists an input state
for which the coherent information $I(A^n\rangle B^n) > 0$ is
positive, and hence $Q(\cN) > 0$.
\end{proof}

\medskip
As usual, quantitative statements about
asymptotically achievable rates and
upper bounds on the identification capacities are naturally
expressed in terms of entropies. For a bipartite density matrix
$\ph^{AB}$, we write
\[H(A)_\ph \equiv H(\ph^A) \equiv -\tr\ph^A\log\ph^A\]
for the \emph{von Neumann entropy} of $\ph^A$.  The mutual
information of the state $\ph^{AB}$ is defined to be
\begin{equation*}
I(A:B)_\ph = H(A)_\ph + H(B)_\ph - H(AB)_\ph
\end{equation*}
while the coherent information and the conditional entropy are, respectively,
\begin{align*}
    I(A\rangle B)_\ph
    &= H(B)_\ph - H(AB)_\ph \\
    H(A|B)_\ph
    &= H(AB)_\ph - H(B)_\ph.
\end{align*}
Our main theorem on the quantum identification capacities includes a
concise formula for $Q_{\rm ID}^{\rm am}$ that eliminates the
optimization over multiple channel uses.
%
%
\begin{theorem}[Quantum identification capacity] \label{thm:Q.ID.capacities}
For any quantum channel $\cN$, its quantum-ID capacity is given by $Q_{\rm
ID}(\cN) = \sup_n \frac{1}{n}Q_{\rm ID}^{(1)}(\cN^{\ox n})$, where
\[
    Q_{\rm ID}^{(1)}(\cN) = \sup_{\ket{\ph}}
        \bigl\{ I(A:B)_\r \text{ s.t. } I(A \rangle B)_\r > 0 \bigr\},
\]
where $\ket{\ph}$ is the purification of any input state to $\cN$
and $\rho^{AB} = (\id\ox\cN)\ph$, and where we declare the $\sup$ to
be $0$ if the set above is empty.

\smallskip \noindent
Furthermore, the amortized quantum-ID capacity equals
\[
    Q_{\rm ID}^{\rm am}(\cN) = \sup_{\ket{\ph}} I(A:B)_\r = C_E(\cN),
\]
the entanglement-assisted classical capacity of $\cN$~\cite{BSST}.
\end{theorem}

\medskip
\begin{remark}
It follows from Theorem \ref{thm:Q.ID.capacities} that the amortized
quantum-ID capacity of a noiseless cbit channel is one. Reconciling
this observation with Theorem \ref{thm:Q.ID.zero}, which asserts
this channel's unamortized quantum-ID capacity is zero, reveals that
\emph{some} amortized noiseless quantum communication is necessary
to achieve $Q_{\rm ID}^{\rm am}$, without determining how much. In
fact, inspection of the proof of Theorem \ref{thm:Q.ID.capacities}
reveals that, for the noiseless cbit channel $\overline\id_2$, a
zero rate of noiseless side qubits is sufficient to achieve the
maximum value of one. These observations extend to cq-channels, so
named because they consist of a destructive measurement resulting in
classical information, followed by the preparation of a state
conditioned on the measurement outcome. For these channels, the
entanglement-assisted capacity $C_E$ is equal to the unassisted
classical capacity $C$, also known as the Holevo
capacity~\cite{Holevo98,SW97}. As a result, $Q_{\rm ID}(\cN)=0$ for
all such channels  even as $Q_{\rm ID}^{\rm am}(\cN) = C(\cN)$, the
latter strictly positive for all nontrivial channels. The difference
in all cases can be traced to a sublinear amount of free quantum
communication in the amortized setting.

This effect can be viewed as an instance of \emph{(un-)locking}
since the quantum-ID rate increases from strictly $0$ to an
arbitrarily large amount by the addition of any positive rate of
quantum communication,
cf.~\cite{I-locking,E-locking,ChristandlWinter}. Unlike the
previously known examples where a certain finite rate is always
required, however, here an arbitrarily small rate of extra quantum
communication is sufficient to bring about an unbounded increase in
the capacity.
\end{remark}

\medskip
The intuition behind the achievability of the rates in
Theorem~\ref{thm:Q.ID.capacities} is quite simple. The structure of
an amortized code is illustrated in Figure~\ref{fig:qid.code}. Fix a
state $\ket{\ph}$ purifying any input to the channel $\cN$ and let
$\ket{\r}^{ABE}$ be $(\1 \ox U_{\cN})\ket{\ph}$, where $U_{\cN}$ is
the Stinespring extension of $\cN$. The encoding will embed the
input into a random subspace of a typical subspace of $A^n$ tensored with
ancillary spaces $C$ and $F$, where $C$ will consist of the amortized
quantum communication and $F$ the environment for the encoding.
Since the encoding is into a random subspace, it will
produce states highly entangled between $B^n C$ and $E^n F$. By
arranging for $E^n F$ to be slightly smaller than $B^n C$ in the
appropriate sense, one ensures that the states are indistinguishable
on the environment $E^n F$. By the weak decoupling duality, they can therefore
be identified by Bob. Letting $R = \smfrac{1}{n} \log |C|$ and $f
= \smfrac{1}{n} \log |F|$, the condition ensuring that $E^n F$ be
``smaller'' than $B^n C$ is roughly
\begin{equation*}
    H(E)_\r + f < H(B)_\r + R,
\end{equation*}
so $f - R$ is chosen to be very slightly less than $H(B)_\r -
H(E)_\r$. Moreover, measure concentration for the choice of random
subspace will make it possible to choose the coding subspace almost
as large as the ambient space $A^n C F$, which in qubit terms has effective
size
\begin{align*}
    n H(A)_\r + nR + nf.
\end{align*}
The rate of the amortized code will therefore be
\begin{align*}
    H(A)_\r + R + f - 2R
    &= H(A)_\r + f - R \\
    &\approx H(A)_\r + H(B)_\r - H(E)_\r.
\end{align*}
Since $\r$ is pure, $H(E)_\r = H(AB)_\r$ which means that the rate is precisely
the mutual information.

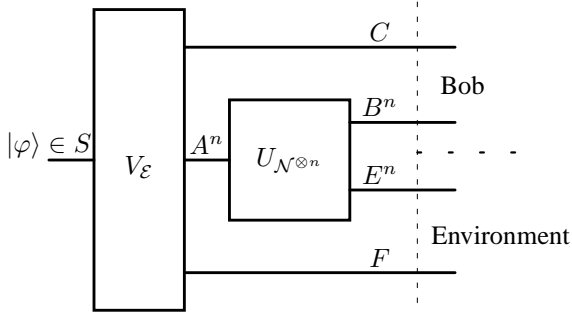
\begin{figure}[ht]
\begin{center}
\def\JPicScale{0.8}
\ifx\JPicScale\undefined\def\JPicScale{1}\fi
\unitlength \JPicScale mm
\begin{picture}(77.5,51.25)(0,0)
\Thicklines
\path(7.5,50)(22.5,50)(22.5,0)(7.5,0)(7.5,50)

\Thicklines
\path(30,35)(50,35)(50,15)(30,15)(30,35)

\put(68.75,37.5){\makebox(0,0)[cc]{Bob}}

\put(75,12.5){\makebox(0,0)[cc]{Environment}}

\Thicklines
\path(7.5,25)(0,25)

\Thicklines
\path(22.5,25)(30,25)

\Thicklines
\path(50,31.25)(67.5,31.25)

\Thicklines
\path(50,20)(67.5,20)

\Thicklines
\path(22.5,43.75)(67.5,43.75)

\Thicklines
\path(22.5,6.25)(67.5,6.25)

\put(40,25){\makebox(0,0)[cc]{$U_{\cN^{\ox n}}$}}

\put(15,23.75){\makebox(0,0)[cc]{$V_{\cE}$}}

\put(55,33.75){\makebox(0,0)[cc]{$B^n$}}

\put(55,22.5){\makebox(0,0)[cc]{$E^n$}}

\put(26.25,27.5){\makebox(0,0)[cc]{$A^n$}}

\put(55,46.25){\makebox(0,0)[cc]{$C$}}

\put(55,8.75){\makebox(0,0)[cc]{$F$}}

\thicklines
\dashline{1}(61.25,26.25)(77.5,26.25)

\thinlines
\dashline{1}(61.25,51.25)(61.25,1.25)

\put(0,27.5){\makebox(0,0)[cc]{$\ket{\ph} \in S$}}

\end{picture}
\end{center}
\caption{Structure of a quantum-ID code. $U_{\cN^{\ox
n}}$ and $V_{\cE}$ are the Stinespring extensions of the noisy
channel $\cN^{\ox n}$ and the encoding operation $\cE$. The
receiver, Bob, has access to the channel output $B^n$ as well as
$C$, which consists of $nR$ qubits transmitted noiselessly from the
receiver. (In the non-amortized setting, there is no $C$.) The
encoding map $\cE$ is generally noisy, so part of its output is
transmitted to the environment as $F$.} \label{fig:qid.code}
\end{figure}

The detailed proof of the achievability of the rates in
Theorem~\ref{thm:Q.ID.capacities} builds on the techniques developed
in Refs.~\cite{generic} and \cite{PopescuShortWinter:thermo}
analyzing the properties of generic quantum states. The proof will
combine the following theorem, originally motivated by the
foundations of statistical mechanics, with the duality between
quantum identification and approximate forgetfulness formulated in
Theorem \ref{thm:ID.dualto.rand} or, more precisely, its technical
variant Theorem \ref{thm:ID.dualto.rand.tech}.
%
%
\begin{theorem}[Random versus average states~\cite{PopescuShortWinter:thermo}] \label{thm:stat.mech}
Let $S$ be a subspace of $B \otimes E$, $\O$ be the maximally mixed
state on $S$, and $X$ any operator acting on $B \otimes E$ with $\|
X \|_\infty \leq 1$. If $\ket{\ph} \in S$ is chosen according to the
unitarily invariant measure, then for all $\e > 0$
\begin{equation*}
    \Pr\left\{ \left\|
         \tr_B X \O X^\dg - \tr_B X \ph X^\dg
    \right\|_1 \geq \eta \right\}
    \leq \eta'
\end{equation*}
where
\begin{align*}
    \eta
    &= \e + \sqrt{\tilde{d}_E / \tilde{d}_B } \quad \mbox{and} \\
    \eta'
    &= 2 \exp( -C \e^2 |S| ).
\end{align*}
Here $C > 0$ is a constant,  $\tilde{d}_E = |\operatorname{supp}
\tr_B XX^\dg|$ is an upper bound on the dimension of the support of
$\tr_B X \O X^\dg$ and $\tilde{d}_B = 1/\tr[ (\tr_E X \O X^\dg)^2 ]$
can be thought of as the effective dimension of $B$.
\end{theorem}
\begin{proof}
This is a slight modification of~\cite[Thm.~2]{PopescuShortWinter:thermo}.
In the original, the theorem
bounds $\| \tr_B \O - \tr_B \ph\|_1$ under similar hypotheses but
$\eta$ includes a correction dependent on $\tr X \Omega X^\dg$. The
correction disappears if the argument is applied to $\| \tr_B X \O
X^\dg - \tr_B X \ph X^\dg \|_1$ instead under the assumption that
$\|X\|_\infty \leq 1$, which ensures that the map $\rho \mapsto X
\rho X^\dg$ is 1-Lipschitz.
\end{proof}

\medskip
In order to use Theorem~\ref{thm:stat.mech} to make statements about
random subspaces, we will use the following lemma
%
%
\begin{lemma} \label{lem:random.subspace}
Let $f$ be a real-valued function on  $\CC P^d$ (identified with
rank one projectors acting on $\CC^d$) and suppose that $f$ is
$\a$-Lipschitz with respect to the trace norm. Let $\mu$ be the
unitarily invariant measure on $\CC P^d$ and $\hat{\mu}$ the
unitarily invariant measure on the space of $k$-dimensional
subspaces of $\CC^d$. If
\begin{equation*}
    \mu\left\{ \ket{\xi} ; f(\xi) > \eta \right\} \leq g(d)
\end{equation*}
then
\begin{equation*}
    \hat\mu\left\{ S ;
        \max_{\ket{\xi} \in S, \braket{\xi}{\xi}=1} f(\xi) > (1+\a) \eta \right\}
    \leq \left( \frac{5}{\eta} \right)^{2k} g(d).
\end{equation*}
\end{lemma}
\begin{proof}
This is a standard discretization argument. Fix a $k$-dimensional
subspace $S_0 \subseteq \CC^d$. According to Ref.~\cite{rand}, there
is a trace norm $\eta$-net $M$ for the rank one projectors on $S_0$
of cardinality no more than $(5/\eta)^{2k}$. If $U$ is distributed
according to the Haar measure $\nu$, then $U S_0$ is distributed
according to the unitarily invariant measure. So, we find by the
triangle inequality that
\begin{align*}
    &\!\!\!\!
     \hat\mu\left\{ S ; \max_{\ket{\xi} \in S, \braket{\xi}{\xi}=1}
        f(\xi) > (1+\a) \eta \right\} \\
    &= \nu\left\{ U; \max_{\ket{\xi} \in S_0, \braket{\xi}{\xi}=1}
        f(U\xi U^\dg) > (1+\a) \eta \right\}  \\
    &\leq \nu\left\{ U;  \max_{\ket{\xi} \in M, \braket{\xi}{\xi}=1}
        f(U\xi U^\dg) >  \eta \right\} \\
    &\leq  \left( \frac{5}{\eta} \right)^{2k}
        \mu\{ \ket{\xi} ; f(\xi) > \eta \},
\end{align*}
where the second inequality is just the union bound over elements of
the net.
\end{proof}
%
%
%

\medskip
The following theorem collects the facts we will need about type and
typical projectors. We omit their definitions, which will not be
required here and can be found in Ref.~\cite{fqsw}.
%
%
\begin{theorem}[Typicality] \label{thm:typicality}
Let $\ket{\rho} \in A \ox B \ox E$ and set $\ket{\ps} =
\ket{\rho}^{\ox n}$. For any $\d, \e > 0$ sufficiently small there
exist projectors $\Pi^B$, $\Pi_1^E$ and $\Pi_2^E$ on $B^{\ox n}$ and
$E^{\ox n}$, respectively, and a projection $\Pi_t^A$ onto a fixed
type subspace of $A^n$ such that the states
\begin{align*}
    \ket{\ps_t} &= \frac{\Pi_t^A \ox \1^B \ox \1^E \ket{\ps}}
                    { \sqrt{ \bra{\ps} \Pi_t^A \ox \1^B \ox \1^E \ket{\ps}}}
    \quad \mbox{and} \\
    \ket{\tilde\ps_t} &= \frac{\Pi_t^A \ox \Pi^B \ox \Pi_2^E \Pi_1^E \ket{\ps}}
                    { \sqrt{ \bra{\ps} \Pi_t^A \ox \1^B \ox \1^E \ket{\ps}}}
\end{align*}
satisfy the following conditions for $X = A^n, B^n, E^n$ and
sufficiently large $n$:
\begin{enumerate}
\item $\ps_t^{A^n} = \Pi_t^A / {\rank \Pi_t^A}$.
\item $\| \ps_t - \tilde\ps_t \|_1 \leq \e$.
\item $\tr [ (\tilde\ps_t^X)^2 ] \leq 3(1-3\e)^{-1} 2^{-n[H(X)_\rho - c\d]}$.
\item $2^{n[H(X)_\rho - \d]} \leq \rank \Pi^X \leq 2^{n[H(X)_\rho+\d]}$.
\item The largest eigenvalue of $\tilde\ps_t^{E^n}$ is bounded
    above by $(1-3\e)^{-1} 2^{-n[H(E)_\r -c\d]}$.
\item The ratio of the largest to the smallest nonzero eigenvalue
    of $\tilde\ps_t^{E^n}$ is at most $2^{2n\d}$.
\end{enumerate}
where $\Pi^A$ and $\Pi^E$ should respectively be understood to be
$\Pi_t^A$ and $\Pi_2^E \Pi_1^E$  in property 4, and $c >0$ is a
constant.
\end{theorem}
\begin{proof}
If $\Pi_2^E$ is removed and property 6 omitted, then the theorem is
precisely a result proved in Ref.~\cite{fqsw}, with $\Pi_1^E$ the
typical projector for $\r$ on $E^n$. $\Pi_2^E$ will be a projector
that removes all eigenvalues of the reduced density operator on
$E^n$ below the stated threshold. Let
\begin{equation*}
    \ket{\xi} = \frac{\Pi_t^A \ox \Pi^B \ox \Pi_1^E \ket{\ps}}
                    { \sqrt{ \bra{\ps} \Pi_t^A \ox \1^B \ox \1^E \ket{\ps}}}
\end{equation*}
The largest eigenvalue of $\xi^{E^n}$ is bounded above by
$(1-3\e)^{-1} 2^{-n[H(E)_\r -c\d]}$ according to property 5 as
stated above and the state's rank is at most $2^{n[H(E)_\r+\d]}$ by
property 4. Applying Lemma \ref{lem:little.eig} to the eigenvalues
of $\xi^{E^n}$ reveals that the sum of all eigenvalues less than or
equal to $2^{-2n\d} / \rank \xi^{E^n}$ is at most
\begin{equation*}
    \frac{2^{-2n\d}}{1-3\e} \leq 2^{-n\d}
\end{equation*}
for sufficiently large $n$. We can therefore let $\Pi_2^E$ be the
orthogonal projection onto the direct sum of the eigenspaces of
$\xi^{E^n}$ corresponding to eigenvalues larger than $2^{-2n\d} /
\rank \xi^{E^n}$. Let $\l$ be the largest eigenvalue of $\xi^{E^n}$.
The ratio of the largest to the smallest eigenvalue after the
application of $\Pi_2^E$ will be at most
\begin{equation*}
    \frac{\l}{2^{-2n\d} / \rank \tilde\ps_t^{E^n}}
    \leq \frac{\l}{2^{-2n\d} \cdot \l}
    = 2^{2n\d}.
\end{equation*}
A redefinition of $\e$ completes the proof.
\end{proof}

\medskip
%
%
\begin{proof+}{(Direct coding part of Theorem \ref{thm:Q.ID.capacities})}
The regular and amortized cases can be handled simultaneously. Fix
an input state $\ph$ as in Theorem \ref{thm:Q.ID.capacities}, let
$\ket{\rho}^{ABE}$ be a purification of $(\id \ox \cN)\ph$ and let
$\ket{\ps} = \ket{\rho}^{\ox n}$. To construct the code, we will
need to project $\ps^{A^n}$ to a type subspace having favorable
properties. $\ps_t^{A^nB^n}$ is the Choi-Jamiolkowski state for the
channel $\cN^{\ox n}$ restricted to the type subspace $A_t$ defined
by the projector $\Pi_t^A$. Call this channel $\cN_t$, write $U_t$
for its Stinespring dilation, and consider $\cN_t \ox \id^C \ox
\id^F$. $C$ will play the role of the noiseless channel from Alice
to Bob in the case of the amortized capacity and $F$ will represent
quantum information discarded by Alice at the encoding stage. Our
code will consist of a subspace of $S'$ of $A_t \ox C \ox F$
selected according to the unitarily invariant measure, which then
defines a subspace $S$ of $(B^n \ox C) \ox (E^n \ox F)$. Our aim
will be to show that $S$ is likely to be approximately forgetful for
$E^n \ox F$ when $C$ and $F$ are chosen appropriately, allowing for
an application of Theorem \ref{thm:ID.dualto.rand.tech}.

Let $\O = \ps_t^{B^nE^n}\ox \pi^C \ox \pi^F$ be the image under $U_t
\ox \1^C \ox \1^F$ of the maximally mixed state on $A_t \ox C \ox
F$. (Recall that $\pi^Z$ denotes the maximally mixed state on $Z$.)
Define $\ket{\tilde{\ps}_t}$ as in Theorem~\ref{thm:typicality} and
let $\tilde{\O} = \tilde{\ps}_t^{B^nE^n}\ox \pi^C \ox \pi^F $. Then
\begin{equation*}
    \tilde{\ps}_t^{B^nE^n} = (\Pi^B \ox \Pi_2\Pi_1^E) \ps_t^{E^nB^n}
            (\Pi^B  \ox \Pi_1^E \Pi_2^E)
\end{equation*}
so for $X = \Pi^B \ox \Pi_2^E\Pi_1^E$, Theorem~\ref{thm:stat.mech}
states that a randomly chosen state $\ket{\o}$ in $U_t(A_t) \ox C
\ox F$ will satisfy
\begin{equation*}
    \Pr\left[ \left\| \tilde{\O}^{E^nF} - \tilde\o^{E^nF}  \right\|_1
                \geq \frac{\eta}{2} \right]
    \leq \eta'
\end{equation*}
for $\tilde\o = X \o X^\dg$ and where, for any $\nu > 0$,
\begin{align*}
    \frac{\eta}{2}
    &= \nu + \sqrt{{\rank [ \Pi_2^E\Pi_1^E \ox \1^F]}\cdot
        {\tr [(\tilde\ps_t^{B^n} \ox \pi^C)^2]}}, \\
    \eta' &= 2 \exp\big(-C \nu^2 |A_t \ox C \ox F| \big).
\end{align*}
We will fix $\nu$ to be $\nu = 2^{-3n\d}$. So by
Lemma~\ref{lem:random.subspace}, a random $S$ in $U_t(A_t) \ox C \ox
F$ chosen according to the unitarily invariant measure will satisfy
\begin{multline*}
    \Pr_S \left[ \max_{\ket{\o} \in S}
        \left\| \tilde{\O}^{E^nF} - \tilde{\o}^{E^nF} \right\|_1 \geq \eta \right]  \\
    \leq 2 \left(\frac{10}{\eta}\right)^{2|S|} \exp\big(-C \nu^2 |A_t \ox C \ox F| \big)
\end{multline*}
since the function $\o \mapsto \| \tilde\O^{E^n F} - \tilde\o^{E^n
F} \|_1$ is 1-Lipschitz with respect to the trace norm. For
convenience, let $|F| = 2^{nf}$ and $|C| = 2^{nR}$.  Since
$|A_t| \geq 2^{n[H(A)_\r-\d]}$, choosing $|S|$ to be
$2^{n[H(A)_\r+R+f-8\d]}$ will lead to
\begin{equation} \label{eq:approx.forget.cap}
      \max_{\ket{\o} \in S}
        \left\| \tilde{\O}^{E^nF} - \tilde{\o}^{E^nF} \right\|_1 < \eta
\end{equation}
with high probability for sufficiently large $n$ provided $\eta$
decays  at most exponentially with $n$.

Now let us determine how to choose $f$ and $R$ in order to ensure a
small value for $\eta$. Observe that by properties 3 and 4 in
Theorem~\ref{thm:typicality},
\begin{align*}
    \rank \Pi_2^E\Pi_1^E\Pi_1^E \ox \1^F
        &\leq 2^{n[H(E)_\r +\d + f]}
            \quad \mbox{and}\\
    \tr [(\tilde\ps_t^{B^n }\ox \smfrac{1}{|C|}\1^C)^2]
        &\leq  3 (1-3\e)^{-1} \cdot 2^{-n[H(B)_\r - c\d - R]}.
\end{align*}
Therefore,
\begin{equation*}
\eta \leq \nu + 3 \cdot 2^{n[H(E)_\r - H(B)_\r +f - R +(1+c)\d ]/2}
\end{equation*}
provided $\e$ is chosen smaller than $1/15$.
There are two cases to consider:

 {\it Case 1.}
 First suppose that $I(A \rangle B)_\r > 0$ or, equivalently, that
$H(E)_\r < H(B)_\r$. Under these circumstances, amortization is not
required. Choosing $R=0$ and $f = H(B)_\r - H(E)_\r - (7+c)\d$ leads
to $\eta \leq \nu + 3 \cdot 2^{-3n\d}  \leq 4 \cdot 2^{-3n\d}$. The
rate of the associated code will be
\begin{align*}
    Q
    &=  \frac{1}{n}\log |S| \\
    &= H(A)_\r +R+f-8\d  \\
    &= H(A)_\r+ H(B)_\r - H(E)_\r - (7+c)\d - 8\d \\
    &= I(A:B)_\r -(15+c)\d.
\end{align*}

{\it Case 2.} Now suppose that $I(A\rangle B)_\r \leq 0$ so that
$H(E)_\r \geq H(B)_\r$. In this case we set $R=H(E)_\r - H(B)_\r +
(7+c)\d$ and $f=0$ to again achieve $\eta \leq  4 \cdot 2^{-3n\d}$.
This time, the rate of the code will be
\begin{align*}
    Q
    &= \frac{1}{n} \log |S| - 2R \\
    &= H(A)_\r + R + f - 8\d - 2R \\
    &= H(A)_\r + H(B)_\r - H(E)_\r - (15+c)\d \\
    &= I(A:B)_\r - (15+c)\d.
\end{align*}

We have established that the subspace $S$ corresponds to a code of
the correct rate. Applying Theorem \ref{thm:ID.dualto.rand.tech} to
$\tilde\O$ and the states in $S$ with $X = \Pi^B \ox \1^C \ox
\Pi_2^E \Pi_1^E \ox \1^F$ will complete the proof. Recalling that
the ratio of the largest to the smallest nonzero eigenvalues of
$\tilde\O^{E^nF}$ is at most $2^{2n\d}$, the theorem asserts that
$S$ is a quantum-ID code with error probability at most
\begin{equation*}
    3 \left( 30 \cdot 2^{2n\d} \cdot (4 \cdot 2^{-3n\d} ) + 4 \sqrt{\e} \right)^{1/2},
\end{equation*}
which can be made arbitrarily small for sufficiently large $n$.
\end{proof+}

\medskip
%
%
\begin{proof+}{(Converse for Theorem \ref{thm:Q.ID.capacities})}
We will address the regular and amortized capacities at the
same time. Consider an amortized quantum-ID code for $n$ copies of
$\cN$ as illustrated in Figure~\ref{fig:qid.code}. The Stinespring
dilations of $\cN^{\ox n}$ and $\cE$ together have three output
registers: one for the channel input, one for the transmission to
Bob and one going to the environment. Abbreviating $\widehat{B} = B^n C$
and $\widehat{E} = E^n F$ in Figure~\ref{fig:qid.code},
the quantum-ID code is
equivalent to a subspace $S \subseteq \widehat{B} \ox \widehat{E}$,
and we can apply our lemmas.

A key observation is that for any pure state ensemble $\{ p_x,\ph_x \}$ 
on $S$ decomposing the maximally mixed state,
\begin{equation}
  \label{eq:entropy-relation}
  H(\widehat{B}) \geq H(\widehat{B}|X) = H(\widehat{E}|X) = H(\widehat{E}) - o(n).
\end{equation}
The first inequality is just the concavity of the entropy function
while the first equality follows from the fact that $\ph_x$ is pure
on $\widehat{B}\widehat{E}$. The final relation is a consequence of
Theorem~\ref{thm:ID.dualto.rand}: the weak decoupling duality implies
that if states can be identified on $\widehat{B}$ then they must be
indistinguishable on $\widehat{E}$. Continuity of the von Neumann
entropy in the form of the Fannes inequality~\cite{Fannes} shows the
correction to be $o(n)$. Thus, sending one half of a maximally
entangled state $\Phi^{AS}$ between $S$ and an auxiliary space
named $A$ into the circuit of Figure~\ref{fig:qid.code},
we obtain a multipartite pure 
state $\Psi^{A\widehat{B}\widehat{E}}$ with respect to which
\[\begin{split}
  \log |A| = H(A) &\leq H(A) + H(\widehat{B}) - H(\widehat{E}) + o(n) \\
                  &= I(A:\widehat{B}) + o(n)     \\
                  &=    I(A:B^n) + I(A:C|B^n) + o(n) \\
                  &\leq I(A:B^n) + 2\log |C| + o(n).
\end{split}\]
Therefore, the amortized quantum identification capacity is bounded
above by $\lim_{n\rar\infty} \smfrac{1}{n} g(\cN^{\ox n})$ where
$g(\cN) = \max_{\ket{\ph}} I(A:B)_\r$ for $\r = (\id \ox \cN)\ph$.
It is well-known, however, that $g(\cN^{\ox n}) = n g(\cN)$ so the
limit is not necessary~\cite{BSST}.

On the other hand, in the non-amortized case, $|C| = 1$, and
the rate of the code is bounded above by $\frac{1}{n}I(A:B^n) + o(1)$.
On the other hand, Eq.~(\ref{eq:entropy-relation}) above yields
\begin{equation}
  \label{eq:I_coh_positive}
  I(A\rangle B^n) = I(A\rangle \widehat{B}) = H(\widehat{B}) - H(\widehat{E}) \geq -o(n),
\end{equation}
which is almost what we need, except that the claim of 
Theorem~\ref{thm:Q.ID.capacities} requires strictly
positive coherent information. We will achieve this by 
modifying the input state in such a way that the coherent
information becomes strictly positive and all other
entropic quantities change only by a sublinear amount (in $n$).

To this end, note that if $Q_{\rm ID}(\cN)=0$ there is nothing to prove,
so we shall assume $Q_{\rm ID}(\cN) > 0$, in which case  $Q(\cN) > 0$
by Theorem~\ref{thm:Q.ID.zero}. Hence, fix a $k$ and the purification 
$\ket{\phi}$ of an appropriate input state to $\cN^{\ox k}$,
such that with respect to $\sigma^{A'B^k} = (\id\ox\cN^{\ox k})\phi$,
$I(A'\rangle B^k)_\sigma \geq 1$, and let 
$\ell = \bigl\lceil \max\{0,-I(A\rangle B^n)_\Psi\} \bigr\rceil + 1$.
Hence, considering block length $N = n+k\ell$ and
the input state $\Phi \ox \phi^{\ox\ell}$ to $\cN^{\ox N}$,
resulting in the state
$\omega^{AA'^\ell B^N} = (\id \ox \cN^{\ox N})(\Psi\ox\phi^{\ox\ell})$,
with respect to which we have
\begin{align*}
  \log|A|           &\leq I(A:B^n) + o(n) \leq I(AA'^\ell:B^N) + o(N), \\
  I(AA'\rangle B^N) &\geq 1 > 0.
\end{align*}
As $N=n+o(n)$, this shows indeed that 
$\sup_n \frac{1}{n}Q_{\rm ID}^{(1)}(\cN^{\ox n})$
is an upper bound on all achievable rates.
\end{proof+}
%
%
%

\section{Non-triviality of amortization rates} 
\label{sec:amortization-rate}
It isn't clear from Theorem~\ref{thm:Q.ID.capacities} alone
what amortization rates $\alpha = \frac{1}{n}\log|C|$ are
necessary to achieve the amortized quantum identification
capacity $Q_{\rm ID}^{\rm am}$ of a given channel.
The previous section established that it is in general impossible
to do entirely without amortization, although an asymptotically zero rate may
suffice to close the gap between $Q_{\rm ID}$ and $Q_{\rm ID}^{\rm am}$,
as is the case for the noiseless cbit channel discussed earlier and, by
similar reasoning, for 
all cq-channels, for which
$Q_{\rm ID}^{\rm am} = C_E = C$, the ordinary classical capacity.

\medskip
To exhibit a channel that requires non-zero asymptotic rate of
amortization to achieve $Q_{\rm ID}^{\rm am}$,
we consider the qubit-erasure channel
$\cE_p:\cL(A) \rightarrow \cL(B)$, with $A=\CC^2$ and $B=\CC^3$,
\[
  \cE_p(\rho) = (1-p)\rho+p\proj{\ast},
\]
for $0\leq p \leq 1$. From Theorems~\ref{thm:Q.ID.capacities} 
and~\ref{thm:Q.ID.zero} we find readily:
\begin{align*}
  Q_{\rm ID}(\cE_p) &= \begin{cases}
                         2(1-p) & \text{ for } 0 \leq p < \frac12, \\
                         0      & \text{ for } \frac12 \leq p \leq 1,
                       \end{cases}                                 \\
  Q_{\rm ID}^{\rm am}(\cE_p) &= 2(1-p).
\end{align*}
Furthermore, for $p<\frac12$, no amortization is necessary, because
the maximum quantum mutual information is attained on the
maximally mixed input, for which $I(A\rangle B) = 1-2p > 0$.
On the other hand, the above shows that for $p\geq\frac12$,
\emph{some} amortization is necessary, although
Theorem~\ref{thm:Q.ID.capacities} does not immediately give bounds on
the rate $\alpha$, except that $\alpha = \max\{2p-1,0\}$
is sufficient, and that for  $p=\frac12$ some amortization, albeit at
zero rate, is necessary and sufficient. The situation is
clarified by the following theorem.

\begin{theorem}
\label{thm:main}
To achieve $Q_{\rm ID}^{\rm am}(\cE_p)$
for $\frac{1}{2} \leq p < 1$, an asymptotic 
amortization rate of at least
\[
  \alpha \geq 2p-1
\]
is necessary and sufficient. At $p=\frac12$, zero rate, but positive
amortization is necessary and sufficient; for $p<\frac12$ and $p=1$,
no amortization is required.
\end{theorem}

\medskip
To prepare the ground, let us 
look first at a single use of the erasure channel with $p > \frac12$ and
an input state $\rho$. Then
\[
  I(A:B) = S(A) + I(A\rangle B)
         = S(\rho) + (1-2p)S(\rho).
\]
The coherent information is always negative, except for 
pure $\rho$. In addition, an amortization rate of $(2p-1)S(\rho)$
is sufficient. 

For $n$ uses of the erasure channel, and a general input state
$\rho$ on $A^n$, 
\begin{align}
  \label{eq:I_n}
  I(A^n:B^n)        &= S(A^n) + I(A^n\rangle B^n), \text{ with} \\
  I(A^n\rangle B^n) &= -S(A^n|B^n) \nonumber \\
  \label{eq:Icoh_n}
                    &\!\!\!\!\!\!\!\!\!\!
                     = \sum_{J\subseteq [n]} p^{|J|} (1-p)^{n-|J|} \bigl( S(J^c)-S(J) \bigr),
\end{align}
where $J^c = [n]\setminus J$ and $S(J)$ is a shorthand for $S(A^J)$.
We know already that the right hand side in Eq.~(\ref{eq:Icoh_n}) is
non-positive (assuming $p\geq\frac12$, as we shall do from now on).
And since any noiseless qubit can only increase the coherent
information by at most $1$, while on the other hand for a
quantum-ID code we need positive coherent information,
we obtain that an amortization rate of $\frac{1}{n} S(A^n|B^n)$
is necessary. 
Motivated by the right hand side of Eq.~(\ref{eq:Icoh_n}), we view
$J\subseteq [n]$ as a random variable describing $n$ Bernoulli trials,
with associated probability $p^{|J|} (1-p)^{n-|J|}$, 
so that
\[
  I(A^n\rangle B^n) = \EE\bigl( S(J^c)-S(J) \bigr),
\]
and using $S(J) = S(A^n) - S(J^c) + I(J:J^c)$, 
\begin{equation}
  \label{eq:mutual-info}
  I(A^n:B^n) = \EE\bigl( 2 S(J^c) - I(J:J^c) \bigr).
\end{equation}

Our strategy will be to develop a lower bound on
the conditional entropy $S(A^n|B^n)$ for all $\rho$ such that
$I(A^n:B^n) \geq 2(1-p)n - \epsilon n$. Here, $\epsilon>0$ is
an arbitrarily small asymptotic parameter, which we let
converge to zero as $n\rightarrow\infty$.

\begin{lemma}
\label{lemma:global}
Under the assumption that
$I(A^n:B^n) \geq 2(1-p)n-\epsilon n$, and for any $L \subseteq J^c$,
\begin{align}
  \label{eq:expected-entropies}
  \EE\bigl( |J^c| - S(J^c) \bigr) &\leq \epsilon n, \\
  \label{eq:expected-entropies-sub}
  \EE\bigl( |L| - S(L) \bigr)     &\leq \epsilon n, \\
  \label{eq:expected-difference}
  \EE\bigl( |J^c| - S(J^c) + I(J:J^c) \bigr) &\leq \epsilon n.
\end{align}
\end{lemma}
\begin{proof}
Eq.~(\ref{eq:expected-entropies}) follows from Eq.~(\ref{eq:expected-difference}).
The latter in turn is seen by comparing Eq.~(\ref{eq:mutual-info})
with $\EE |J^c| = (1-p)n$:
\[\begin{split}
  \epsilon n &\geq 2(1-p)n - I(A^n:B^n) \\
             &=    \EE\bigl( 2|J^c| - 2S(J^c) + I(J:J^c) \bigr) \\
             &\geq \EE\bigl( |J^c| - S(J^c) + I(J:J^c) \bigr) . 
\end{split}\]
Finally, Eq.~(\ref{eq:expected-entropies-sub}), follows
by subtracting at most $S(J^c\setminus L) \leq |J^c|-|L|$ from
$S(J^c)$ in Eq.~(\ref{eq:expected-entropies}), and taking expectations.
\end{proof}

\medskip
The above says that for typical $J$, the entropies $S(J^c)$
are $\approx (1-p)n$, which is almost as large as they can
be, since with high probability, $|J^c| \approx (1-p)n$;
furthermore, $S(J)$ must be of the same order, and the
mutual information $I(J:J^c)$ between blocks $J^c$ and $J$
is small. However, $J$ is typically larger than $J^c$
(being of size $pn$ and $(1-p)n$, respectively), so we can
find several random $J_i^c$ in $J$, which will result in a lower
bound on the entropy of $J$. 

\begin{lemma}
\label{lemma:local}
Consider a random subset $J \subseteq [n]$ distributed
according to $p^{|J|} (1-p)^{n-|J|}$, and
let $k=\left\lceil \frac{p}{1-p} \right\rceil$.
Then, for sufficiently large $n$,
\[
  \EE S(J) \geq pn - k\epsilon n - 1.
\]
\end{lemma}
\begin{proof}
Define random subsets $J_1,\ldots,J_k \subseteq [n]$ with the following
distribution: for $|J| < \frac12 n$, let them be independent
and uniformly chosen from the subsets of size $|J|$, and for
$|J| \geq \frac12 n$ choose $K_1,\ldots,K_k \subseteq J$ with
$|K_i|=|J^c|$ such that $\left| \bigcup_i K_i \right|$ is as
large as possible (i.e.~either $|J|$ or $k|J^c|$, whichever
is smaller); then let $J_i^c := \pi(K_i)$ for a uniformly
random permutation $\pi$ of $J$.

Note that each $J_i$ has the same Bernoulli distribution as $J$,
but that the complements $J^c,J_1^c,\ldots,J_k^c$ are ``as disjoint
as possible.'' 

Now,
\[\begin{split}
  S(J) &=    S(J\cap J_1^c) + S(J\cap J_1) - I(J_1^c:J\cap J_1) \\
       &\geq \bigl( S(J_1^c) - I(J_1:J_1^c) \bigr) + S(J\cap J_1) \\
       &\geq \bigl( S(J_1^c) - I(J_1:J_1^c) \bigr)  \\
       &\phantom{=}
              + \bigl( S(J \cap J_1 \cap J_2^c) - I(J_2:J_2^c) \bigr) + S(J \cap J_1 \cap J_2) \\
       &\geq \sum_{i=1}^k \bigl( S(J_i^c \cap J \cap J_1 \cap \ldots \cap J_{i-1}) - I(J_i:J_i^c) \bigr),
\end{split}\]
and taking expectations, using Eq.~(\ref{eq:expected-difference}),
results in
\[
  \EE S(J) \geq \EE\left| J \cap \bigcup_{i=1}^k J_i^c \right| - k\epsilon n.
\]
Noting that $\bigcup_{i=1}^k J_i^c = J$ 
except with exponentially small probability (in $n$), the claim follows for sufficiently
large $n$.
\end{proof}

\medskip
\begin{proof+}{of Theorem~\ref{thm:main}}
We simply put together the bounds in Lemmas~\ref{lemma:global}
and~\ref{lemma:local}:
\[\begin{split}
  \EE\bigl( S(J^c)-S(J) \bigr) &\leq \EE|J^c| - \EE S(J) \\
                               &\leq n(1-p) - pn + k\epsilon n + 1,
\end{split}\]
with $k=\left\lceil \frac{p}{1-p} \right\rceil$.
For $n\rightarrow \infty$ and $\epsilon\rightarrow 0$ this yields
the claim.

The other parts of the theorem we knew already.
\end{proof+}

\medskip
\begin{remark}
It is interesting to note that as $p\rightarrow 1$, while the
capacity $2(1-p) \rightarrow 0$, the amortization cost
$2p-1 \rightarrow 1$, despite the fact that at $p=1$, the
capacity and naturally also the amortization cost are zero.
The minimal amortization cost required to achieve $Q_{\rm ID}^{\rm am}$ 
is therefore not a continuous function of the channel.
\end{remark}

\medskip
The above method can even be applied to prove
bounds on the rate-amortizaton tradeoff. Note that because the
classical capacity of $\cE_p$ is $C(\cE_p) = 1-p$, at that rate
of identification we can do with amortization at zero rate.
It seems reasonable to conjecture that for quantum
identification rates between $1-p$ and $2(1-p)$, a strictly positive
amortization rate is necessary.

\begin{figure}[ht]
  \includegraphics[width=8.5cm]{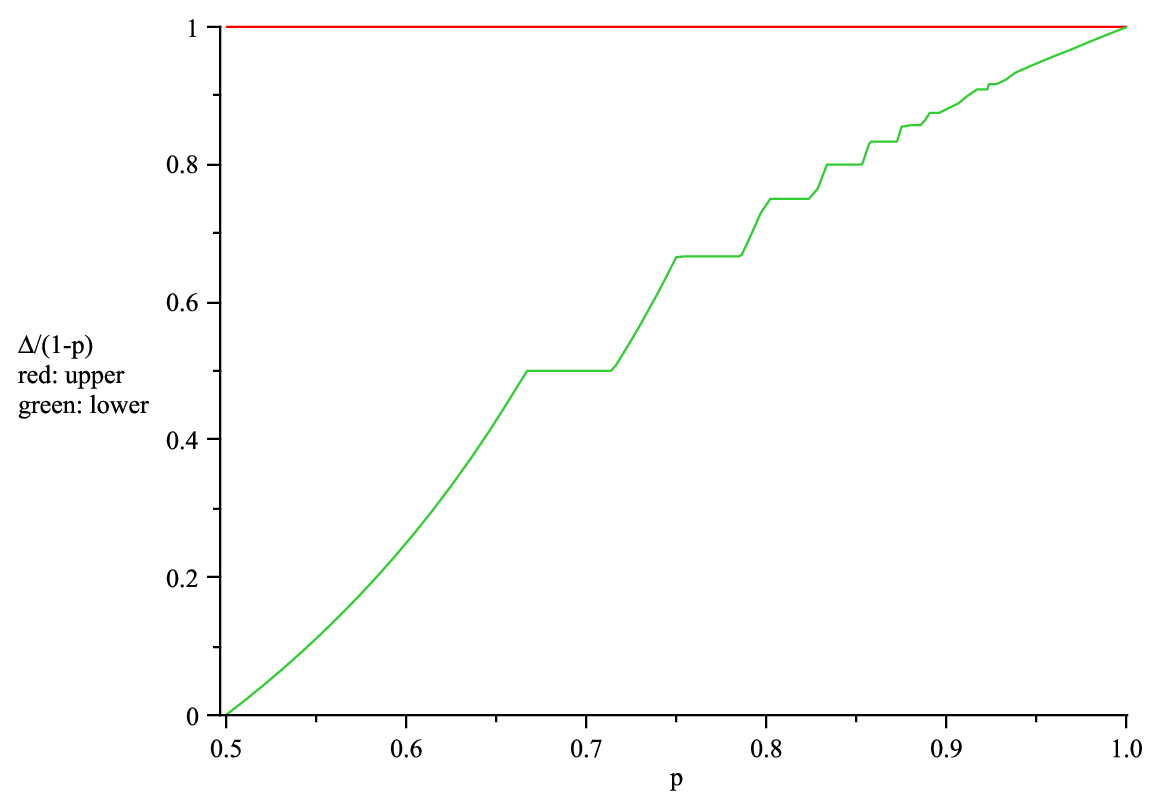}
  \caption{Plot of $p$ versus $\frac{\Delta}{1-p}$ for the minimum $\Delta$
      so that the amortization rate $\alpha=0$: the red line
      (constant $1$) corresponds to the bound mentioned above that at
      rate $R=1-p$ zero amortization rate is sufficient. The green
      plot is the lower bound on $\Delta$ from Theorem~\ref{thm:tradeoff}.}
\end{figure}

\begin{theorem}
\label{thm:tradeoff}
To achieve an amortized quantum identification rate $R=2(1-p)-\Delta$
asymptotically, for $\frac{1}{2} \leq p < 1$, an amortization rate of
at least
\[\begin{split}
  \alpha &\geq 2p-1 - \left\lceil \frac{p}{1-p} \right\rceil \Delta \\
         &\geq 2p-1 - \frac{\Delta}{1-p}
\end{split}\]
is necessary. Another, sometimes better, lower bound is
\[\begin{split}
  \alpha &\geq \left\lfloor \frac{2p-1}{1-p} \right\rfloor (1-p-\Delta) - \Delta \\
         &\geq 1-p-2\Delta\quad \text{ (for }p\geq 2/3).
\end{split}\]
\end{theorem}
\begin{proof}
This parallels the proof of Theorem~\ref{thm:main}, except that
Eq.~(\ref{eq:expected-difference}) is replaced by
\[
  \EE\bigl( |J^c| - S(J^c) + I(J:J^c) \bigr) \leq (\Delta+\epsilon)n.
\]
This implies, for $k < \left\lceil \frac{p}{1-p} \right\rceil$,
\[
  \EE S(J) \geq k(1-p)n - k\Delta n - k\epsilon n - 1,
\]
and for $k = \left\lceil \frac{p}{1-p} \right\rceil$,
\[
  \EE S(J) \geq pn - k\Delta n - k\epsilon n - 1.
\]
The rest of the argument is the same.
\end{proof}

\section{Conclusion and open questions} \label{sec:conclusions}

Weak decoupling duality is the statement that geometry preservation and
approximate forgetfulness are complementary properties, much like
quantum data transmission and complete forgetfulness. Subject to
some technical conditions, geometry preservation is itself
equivalent to quantum identification, an operational task very much
in the spirit of quantum data transmission but strictly weaker. Just
as analyzing complete forgetfulness has proved a versatile and
effective tool for studying asymptotic quantum error correction,
approximate forgetfulness provides a new approach to asymptotic
quantum identification. Indeed, by focusing on approximate
forgetfulness of the complementary channel, we have established that
the amortized quantum identification capacity is exactly equal to
the entanglement-assisted capacity.

The weak decoupling duality suggests a number of possible extensions,
such as asking what happens if geometry is preserved not only for
pure states but for higher rank mixed states. Would such a property
have an operational interpretation and corresponding interpretation
in terms of a form of forgetfulness intermediate between the weak
form studied here and complete forgetfulness? It would also be
interesting to understand geometry preservation as a type of
pseudo-isometry~\cite{mostow} from projective space to the
Grassmannian of subspaces corresponding to the supports of the mixed
output states. 

Meanwhile, Theorem~\ref{thm:Q.ID.capacities} poses an entertaining
and potentially deep puzzle: why do amortized quantum identification
and entanglement-assisted classical communication result in the same
capacity in the absence of any known operational relationship
between these tasks? The theorem also leaves open the important
problem of evaluating the quantum identification capacity formula in
the unamortized case (we expect that to be difficult as it includes
deciding whether the quantum capacity is positive).
We also left open precisely how much amortized quantum communication 
is necessary to achieve the amortized capacity, 
although we were able to determine the optimal amortization rate 
in the case of an erasure channel, showing that it is
strictly positive for erasure probability larger than $\frac12$.
More generally, the nature of the tradeoff between achievable 
identification rates and amortization rates is completely unknown.

\section*{Acknowledgments}
We thank Marco Piani for sharing his proof of
Lemma~\ref{lem:norm.bound} with us and Mark Wilde for helpful comments.

\appendix


The following results were used in various proofs but have been collected
here so as not to distract from the main line of argument in the paper.
This first relation provides a convenient way to
calculate mixed state fidelity:
%
%
\begin{lemma} \label{lemnice}
For pure states $\ph,\psi$ on a bipartite system $B\ox E$,
\begin{equation} \label{eq:neat.identity}
    F(\ph^B,\psi^B) = \bigl\| \tr_B \ket{\ph}\!\bra{\psi} \bigr\|_1^2 .
\end{equation}
\end{lemma}
\begin{proof}
This is a straightforward calculation:
\begin{equation*}\begin{split}
    \bigl\| \tr_B \ket{\ph}\!\bra{\psi} \bigr\|_1
         &= \max_{\| X \|_\infty \leq 1}
            \left| \tr\left(\tr_B \ket{\ph}\!\bra{\psi}\right)X \right| \\
         &= \max_{U \text{ unitary}}
            \left| \tr\left(\tr_B \ket{\ph}\!\bra{\psi}\right)U \right| \\
         &= \max_{U \text{ unitary}}
            \left| \tr\ket{\ph}\!\bra{\psi}(\1\ox U) \right| \\
         &= \max_{U \text{ unitary}}
            \sqrt{ F\bigl( (\1\ox U)\ph(\1\ox U^\dagger),\psi \bigr) } \\
         &= \sqrt{F(\ph^B,\psi^B)},
\end{split}\end{equation*}
invoking, successively, the duality between trace and sup norm, the
fact that the maximum is always attained at a unitary, the defining
property of the partial trace, and in the last line Uhlmann's
relation~\cite{jozsa:fid,uhlmann:fid}.
\end{proof}

\medskip
The following lemma provides conditions under which mixing preserves
near-orthogonality.
%
%
\begin{lemma} \label{lemmixing.vs.fidelity}
Let $\rho$ and $\sigma_i$, for all $i$, be states on the same
Hilbert space such that there exist projectors $P$ and $Q_i$ of rank
$\leq r$, and $\mu P \leq \rho \leq \lambda P$, $\mu Q_i \leq
\sigma_i \leq \lambda Q_i$ such that $\mu r \leq 1$.
If furthermore for all $i$, $F(\rho,\sigma_i) \leq \epsilon$, then
\[
    F\bigl( \rho,\overline\sigma \bigr) \leq \delta := \epsilon \frac{\l^2}{\m^2}
\]
for every $\overline\sigma = \sum_i p_i\sigma_i$ in the convex hull
of the $\sigma_i$.
\end{lemma}
\begin{proof}
We use the definition of the fidelity to first obtain
\[
    \epsilon \geq \left( \tr\sqrt{\sqrt{\rho}\sigma_i\sqrt{\rho}} \right)^2
             \geq \mu^2 \left( \tr PQ_iP \right)^2.
\]
Invoking the definition again, we now get from this
\[\begin{split}
    \sqrt{F\bigl( \rho,\overline\sigma \bigr)}
            = \left\| \sqrt{\rho}\sqrt{\overline\sigma} \right\|_1
                &\leq \lambda \tr\sqrt{\sum_i p_i PQ_iP}                         \\
                &\leq \lambda r \sqrt{\sum_i p_i \frac{1}{\mu r} \mu \tr PQ_iP}  \\
                &\leq \lambda r \sqrt{\frac{\epsilon}{\mu r}}
                 \leq \sqrt{\epsilon}\frac{\lambda}{\mu},
\end{split}\]
using the concavity of the square root twice in turn~\cite{bhatia}.
\end{proof}
%
%
%

%
%
\begin{lemma} \label{lem:fidelity.op.mono}
Let $0 \leq \tilde\r \leq \r$ and $0 \leq \tilde\s \leq \s$. Then
$F(\tilde\r,\tilde\s) \leq F(\r,\s)$.
\end{lemma}
\begin{proof}
Denoting unitary congruence of matrices (in particular
having the same spectrum) by $\sim$, we have
\begin{equation*}
    \sqrt{\tilde\r} \tilde\s \sqrt{\tilde\r}
    \leq \sqrt{\tilde\r} \s \sqrt{\tilde\r}
    \sim \sqrt{\s} \tilde\r \sqrt{\s}
    \leq \sqrt{\s} \r \sqrt{\s}
    \sim \sqrt{\r} \s \sqrt{\r}.
\end{equation*}
Hence, since the square root is operator monotone~\cite{bhatia}
and the trace is invariant under unitary basis change,
$\tr\sqrt{\sqrt{\tilde\r} \tilde\s \sqrt{\tilde\r}} \leq \tr\sqrt{\sqrt{\r} \s \sqrt{\r}}$,
completing the proof.
\end{proof}

\medskip
The next lemma constrains the increase of the maximal output trace
norm when tensoring with a fixed-size identity transformation:
%
%
\begin{lemma} \label{lem:norm.bound}
Let $\G : \cS(A) \rar \cS(B)$ be a linear superoperator. Then for
any $t$ any positive integer,
\begin{equation*}
    \left\| \G \right\|_\diamond^{(t)}
    \leq t \left\| \G \right\|_\diamond^{(1)}.
\end{equation*}
\end{lemma}
\begin{proof}
Write $X$, an operator on $\CC^t \ox A$ such that $\|X\|_1 \leq 1$,
in its singular value decomposition as $\sum_j s_j \ketbra{v_j}{w_j}
$, with $0 \leq s_j \leq 1$ and $\braket{v_j}{v_k} =
\braket{w_j}{w_k} = \d_{jk}$. By convexity (triangle inequality),
$\left\| \G \right\|_\diamond^{(t)}$ is attained with a
rank-one $X = \ketbra{v}{w}$, and for the following fix
Schmidt decompositions $\ket{v} = \sum_k \a_{k} \ket{e_{k}}\ket{f_{k}}$
and $\ket{w} = \sum_\ell \b_{\ell} \ket{g_{\ell}}\ket{h_{\ell}}$. Then,
\begin{align*}
  \left\| (\id_t \ox \G)X \right\|_1
    &=    \bigl\| (\id_t \ox \G)\ketbra{v}{w} \bigr\|_1 \\
    &\!\!\!\!
     =    \left\| (\id_t \ox \G) \!\left(
        \sum_{k\ell} \a_{k} \b_{\ell} \ketbra{e_{k}}{g_{\ell}}
            \ox \ketbra{f_{k}}{h_{\ell}} \right)\! \right\|_1 \\
    &\!\!\!\!
     \leq \sum_{k\ell} \a_{k} \b_{\ell} \bigl\| (\id_t \ox \G)\left(
        \ketbra{e_{k}}{g_{\ell}}
            \ox \ketbra{f_{k}}{h_{\ell}} \right) \bigr\|_1 \\
    &\!\!\!\!
     = \sum_{k\ell} \a_{k} \b_{\ell} \left\| \G\left(
        \ketbra{f_{k}}{h_{\ell}} \right) \right\|_1
     \leq t \| \G \|_1^{(1)},
\end{align*}
where the first step is just the triangle inequality and the next follows
from the fact that $\|X\|_1 = \sum_j s_j \leq 1$. The final
inequality uses the fact that $\sum_{k=1}^t \a_{k}$ and
$\sum_{l=1}^t \b_{j}$ are both bounded above by $\sqrt{t}$ since
$\| \a \|_2 = \| \b \|_2 = 1$.
\end{proof}

\begin{remark}
The factor $t$ is optimal,
as the example of the matrix transposition shows
where the bound of the lemma becomes an equality.
\end{remark}

%
%
\begin{lemma}[Gentle measurement~\cite{winter:q-strong,ogawa:nagaoka,ogawa:nagaoka2}]
  \label{lemma:gentle}
Let $\rho$ be a state, and $0 \leq X \leq \1$ be an operator on some
Hilbert space, such that $\tr \rho X \geq 1-\epsilon$. Then,
\(
    \bigl\| \rho - \sqrt{X}\rho\sqrt{X} \bigr\|_1 \leq 2\sqrt{\epsilon}.
\)
\end{lemma}

The following, final, lemma is used to argue that the small eigenvalues of a
density operator can be discarded without causing much disturbance.
%
%
\begin{lemma} \label{lem:little.eig}
Let $(p_1,p_2,\ldots,p_r)$ be a probability density with $p_i \geq
p_{i+1}$ for all $i$ and let $\chi = \{ i ; p_i \leq D/r \}$ for some
$0 \leq D \leq 1$. Then,
$\sum_{i \in \chi} p_i \leq D$.
\end{lemma}
\begin{proof}
Since evidently $|\chi| \leq r$,
\[
  \sum_{i \in \chi} p_i \leq |\chi|\frac{D}{r} \leq r\frac{D}{r} = D,
\]
and that's it.
\end{proof}

\bibliographystyle{unsrt}

\begin{biography}{Patrick Hayden}
(MÕ04) received his doctorate from the University of Oxford 
in 2001 and was subsequently a postdoctoral fellow at the California 
Institute of Technlogy (Caltech) until 2004. 
He is an associate professor in McGill UniversityÕs School of Computer 
Science and a Distinguished Research Chair of the Perimeter Institute for
Theoretical Physics. His research focuses on quantum information theory 
and its applications to other areas of physics and computer science. 
\end{biography}

\begin{biography}{Andreas Winter}
received the Diploma degree in Mathematics from the Freie 
Universit\"at Berlin, Berlin, Germany, in 1997, and the Ph.D. degree from the 
Fakult\"at f\"ur Mathematik, Universit\"at Bielefeld, Bielefeld, Germany, in 1999. 
He was Research Associate at the University of Bielefeld until 2001, and 
since then, he has been with the University of Bristol, Bristol, U.K., initially 
as Research Associate in computer science, from 2003 as Lecturer in 
Mathematics, and since 2006 as Professor of Physics of Information. In 2007, he also 
became affiliated with the Centre for Quantum Technologies, National Univer- 
sity of Singapore.
\end{biography}

\end{document}